\documentclass[showpacs,amsmath,amssymb,twocolumn,superscriptaddress,notitlepage,preprintnumbers,aps,pra,10pt,floatfix]{revtex4-2}
\usepackage{graphicx}
\usepackage{dcolumn}
\usepackage{bm}
\usepackage{qcircuit}
\usepackage{braket}
\usepackage{color}
\usepackage{amsthm}
\usepackage{chemarrow}
\usepackage{overpic}
\usepackage{diagbox}
\usepackage{braket}
\usepackage{subfigure}
\usepackage{makecell}
\newtheorem{theorem}{Theorem}
\newtheorem{definition}{Definition}
\newtheorem{lemma}{Lemma}
\newtheorem{corollary}{Corollary}
\usepackage[colorlinks=true,linkcolor=red,bookmarks=true,breaklinks=true]{hyperref}
\usepackage{chngcntr}
\begin{document}
\let\oldacl\addcontentsline
\renewcommand{\addcontentsline}[3]{}
\newcommand{\Tr}{\mathrm{Tr}}
\newcommand{\qi}[1]{\textcolor{blue}{(QI: #1)}}
\newcommand{\szx}[1]{\textcolor{red}{(SZX: #1)}}
\title{Entanglement-induced exponential advantage in amplitude estimation via state matrixization}
\author{Zhong-Xia Shang}
\affiliation{HK Institute of Quantum Science $\&$ Technology, The University of Hong Kong, Hong Kong, China}
\affiliation{QICI Quantum Information and Computation Initiative, Department of Computer Science,
The University of Hong Kong, Hong Kong, China}
\affiliation{Shanghai Research Center for Quantum Science and CAS Center for Excellence in Quantum Information and Quantum Physics,
University of Science and Technology of China, Shanghai 201315, China}
\author{Qi Zhao}
\affiliation{QICI Quantum Information and Computation Initiative, Department of Computer Science,
The University of Hong Kong, Hong Kong, China}

\begin{abstract}
Estimating quantum amplitude, or the overlap between two quantum states, is a fundamental task in quantum computing and underpins numerous quantum algorithms.
In this work, we introduce a novel algorithmic framework for quantum amplitude estimation by transforming pure states into their matrix forms (Matrixization) and encoding them into non-diagonal blocks of density operators and diagonal blocks of unitary operators. Utilizing the construction details of state preparation circuits, we systematically reconstruct amplitude estimation algorithms within the novel matrixization framework through a technique known as channel block encoding. Compared with the standard approach, amplitude estimation through matrixization can have a different complexity that depends on the entanglement properties of the two quantum states. Specifically, our new algorithm can have exponentially smaller gate complexity when one of the two quantum states is prepared by a linear-depth quantum circuit that is below maximal entanglement under a certain bi-partition and the other state is maximally entangled. We later generalize this result to broader regimes and discuss implications. Our results demonstrate that the near-optimal performance of the standard amplitude estimation algorithm can be surpassed in specific cases. 
\end{abstract}
\maketitle

\noindent\textit{\textbf{Introduction.—}}
Quantum computing is a critical field with the potential to achieve significant speed-ups for certain computational problems \cite{nielsen2010quantum}. Within this domain, the estimation of quantum amplitudes in the form of the inner product of two distinct quantum states, $\langle B|A\rangle$, 
holds considerable significance. In physics, the Loschmidt echo \cite{goussev2012loschmidt}, a class of quantum amplitudes, plays an important role in understanding quantum dynamics \cite{andraschko2014dynamical}, quantum chaos \cite{yan2020information}, and decoherence \cite{cucchietti2003decoherence}. In quantum computing, quantum amplitudes can encode solutions to classically intractable problems \cite{hangleiter2023computational,aaronson2011computational},  leveraging quadratic quantum speedup \cite{grover1996fast,brassard2002quantum} and forming the basis of quantum sampling advantage tasks \cite{hangleiter2023computational,arute2019quantum,zhong2020quantum}. They are also pivotal in various quantum machine learning tasks \cite{havlivcek2019supervised,lloyd2013quantum,schuld2019quantum}. 

Current quantum algorithms for amplitude estimation can be classified into two main categories. The Hadamard test \cite{aharonov2006polynomial} estimates quantum amplitudes with a standard quantum limit (s.q.l.) estimation ($\epsilon^{-2}$) with a query complexity of $\mathcal{O}(|\mu|^{-2}\epsilon^{-2}\log(\delta^{-1}))$. This method approximates $\mu=\langle B|A\rangle$ with a relative error $\epsilon$ and a success probability no less than $1-\delta$. The amplitude estimation algorithm \cite{brassard2002quantum} and its variants \cite{aaronson2020quantum,grinko2021iterative,rall2023amplitude,giurgica2022low} achieve a quadratic reduction to the Heisenberg limit (h.l.)($\epsilon^{-1}$) with a query complexity $\mathcal{O}(|\mu|^{-1}\epsilon^{-1}\log(\delta^{-1}))$. Here, the query complexity denotes the number of calls to the preparation oracles for $|A\rangle$ and $|B\rangle$. The Heisenberg scaling represents the optimal outcome allowed by the uncertainty principle inherent in the linearity of quantum mechanics
\cite{giovannetti2004quantum,bennett1997strengths,childs2016optimal,abrams1998nonlinear,wootters1982single}. Consequently, whenever $\mu$ is exponentially small, an exponential amount of cost is inevitably required to obtain an accurate estimation, underpinning the general belief that $NP \notin BQP$ \cite{bennett1997strengths,aaronson2005guest}.


The optimality of amplitude estimation is typically considered under the black-box assumption. In this work, we explore whether the complexity of estimating amplitudes can be further reduced if we have knowledge of the operational details of the state preparation unitary operators, such as their construction from elementary gates or engineered Hamiltonian evolutions. We give a surprisingly positive answer to the question. We propose a new algorithmic framework for estimating quantum amplitudes by transforming quantum states into matrices and encoding them into blocks of density matrices and unitary operators. Through our new approach, the complexity of estimating amplitudes can be further \textbf{exponentially} reduced beyond the standard amplitude estimation for certain regimes which interestingly has a strong connection with entanglement properties of the two quantum states $|A\rangle$ and $|B\rangle$.

\noindent\textit{\textbf{Problem setup and main results.—}}
We aim to estimate the real (imaginary) part of the amplitude $\mu=\langle B|A\rangle=\mu_r+i\mu_i$ up to a small relative error $\epsilon$: $|\overline{\mu_{r(i)}}-\mu_{r(i)}|\leq \epsilon |\mu_{r(i)}|$ with a success probability at least $1-\delta$. $|A\rangle$ and $|B\rangle$ are $2n$-qubit states, and we designate the first $n$-qubit subsystem as the upper subsystem (US) and the last $n$-qubit subsystem as the lower subsystem (LS). For $|A\rangle$, which is unknown, we have $|A\rangle=V|A_0\rangle$ with $|A_0\rangle$ an initial product state and $V$ a quantum circuit with known construction details from the elementary gate set $\{H,S,T,CNOT\}$. For $|B\rangle$, it is a known $2n$-qubit Bell basis state, which is a tensor product of 2-qubit Bell states on the pairs ${(1, n+1), (2, n+2), \ldots, (n, 2n)}$. 

In this work, we propose a new quantum algorithm to estimate such amplitudes with the following main theorem (also see Fig.~\ref{f1}a).
\begin{theorem}\label{cbec}
For the amplitude $\mu=\langle B|A\rangle=\mu_r+i\mu_i$, when the gate complexity in $V$ is $\tau$ with $K$ CNOT gates connecting US and LS, its estimation to a relative error $\epsilon$ with a success probability $1-\delta$ can be achieved with a gate (time) complexity
\begin{eqnarray}
T_{s}=\mathcal{O}(\tau 2^{K-n+2} |\mu_{r(i)}|^{-2}\epsilon^{-2} \log(\delta^{-1})),
\end{eqnarray}
for s.q.l. estimation and
\begin{eqnarray}
T_{h}=\mathcal{O}(\tau 2^{\frac{K-n}{2}+1}|\mu_{r(i)}|^{-1}\epsilon^{-1} \log(\delta^{-1})),
\end{eqnarray}
for h.l. estimation.
\end{theorem}
\noindent Compared with standard methods, we have the additional factor $2^{\frac{K-n}{2}+1}$, therefore, whenever $n-K=\Theta(n)$ which is true for $K=n/2 
 \text{(linear depth)}, \sqrt{n}$,  etc, our new algorithm can give \textbf{exponential} improvements. We will show later that this regime is classically hard in various aspects. Thus, our algorithm exponentially improves the amplitude estimation in a quantum advantage regime. Our result can also be generalized to other forms of $|B\rangle$ with milder improvements (see additional results and SM \ref{ape8}).

\noindent\textit{\textbf{Methods overview.—}}
The $2n$-qubit states $|A\rangle$ and $|B\rangle$ have the expressions $|A\rangle=\sum_{ij}\alpha_{ij}|i\rangle|j\rangle$ and $|B\rangle=\sum_{ij}\beta_{ij}|i\rangle|j\rangle$ with $|i\rangle$ in US and $|j\rangle$ in LS. We first turn them into $n$-qubit matrices $\mathcal{M}[|A\rangle]=A=\sum_{ij}\alpha_{ij}|i\rangle\langle j|$ and $\mathcal{M}[|B\rangle]=B=\sum_{ij}\beta_{ij}|i\rangle\langle j|$, a process we call the matrixization procedure. We will also call the inverse procedure $\mathcal{V}[\cdot]$ the vectorization procedure. We can also understand these two mappings from Choi–Jamiołkowski isomorphism \cite{choi1975completely} (See SM \ref{overview}). 

\begin{figure}[htbp]
\centering
\includegraphics[width=0.49\textwidth]{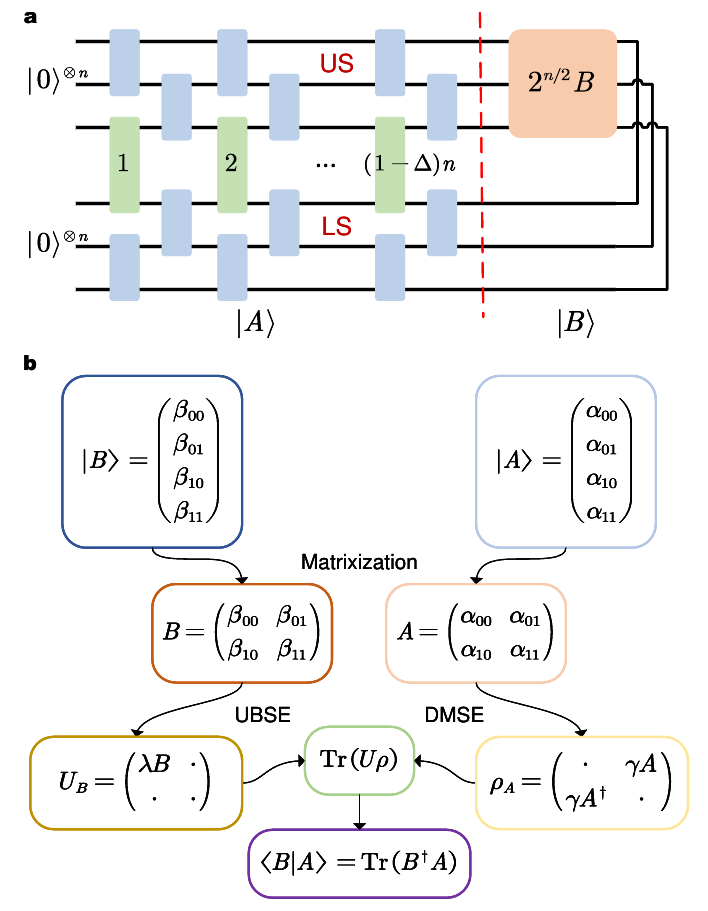}
\caption{(a): The amplitudes that the new algorithm can give exponential improvement (Theorem \ref{cbec}). $(1-\Delta)n$ is to denote $n-K=\Theta(n)$. (b): Sketch of the algorithm. Matrices $A$ and $B$ are matrixizations of $|A\rangle$ and $|B\rangle$. We encode $A$ into a non-diagonal block of a density matrix $\rho_A$ by DMSE and use UBSE to let the unitary operator $U_B$ be a block encoding of $B$. In this way, we can further use the Hadamard test or the amplitude estimation to evaluate values of the form $\Tr(U\rho)$ to estimate the amplitude $\langle B|A\rangle=\Tr(B^\dag A)$.\label{f1}}
\end{figure}

When $|A\rangle$ and $|B\rangle$ are turned into matrices, we have $\langle B|A\rangle=\Tr(B^\dag A)$. Thus, we can find quantum objects to encode $A$ and $B$ such that the value of $\Tr(B^\dagger A)$ can be estimated. We propose unitary block state encoding (UBSE) and density matrix state encoding (DMSE):
\begin{definition}[Unitary block state encoding (UBSE)]
Given an $n+k$-qubit unitary operator $U_B$, if $U_B$ satisfies
\begin{eqnarray}
\mathcal{V}[(I_n\otimes \langle 0|_k ) U_B (I_n\otimes |0\rangle_k)]=\lambda |B\rangle,
\end{eqnarray}
then $U_B$ is called an $(n+k,\lambda)$-UBSE of $|B\rangle$.
\end{definition}
\noindent In other words, $U_B$ is a block encoding \cite{low2019hamiltonian} of $\mathcal{M}[|B\rangle]$. For the special forms of $|B\rangle$ we considered, $B$ can be exponentially amplified ($\lambda=2^{n/2}$) and encoded into a Pauli operator $U_B=2^{n/2} B$, which is the key resource of the exponential improvements of our algorithm (See SM \ref{bbss}). Meanwhile, DMSE can defined as:
\begin{definition}[Density matrix state encoding (DMSE)]
Given an $l+n$-qubit density matrix $\rho_A$, if $\rho_A$ satisfies
\begin{eqnarray}
\mathcal{V}[(\langle s_1|_l\otimes I_n) \rho_A (|s_2\rangle_l\otimes I_n)]=\gamma |A\rangle,
\end{eqnarray}
where $|s_1\rangle_l$ and $|s_2\rangle_l$ are two orthogonal computational basis states of the $l$-qubit ancilla system, then $\rho_A$ is called an $(l+n,s_1,s_2,\gamma)$-NDSE of $|A\rangle$.
\end{definition}
\noindent The non-diagonal block of density matrices can get rid of the Hermitian and positive semi-definite restrictions of density matrices. In this work, we will mainly focus on $(1+n,0,1,\gamma)$-DMSE where $\rho_A$ has the form:
$\rho_A=\begin{pmatrix}
\cdot & \gamma A \\
\gamma A^\dag & \cdot
\end{pmatrix}$. Note that $\rho_A$ is also a $(1+n,1,0,\gamma)$-DMSE of $|A^\dag\rangle$. 

An important thing is how to prepare $\rho_A$ from the construction information of $V$. A reasonable workflow is to reconstruct $V$ under the matrixization picture and act it on the DMSE of $|A_0\rangle$ defined as $\rho_{A_0}$ (with $\gamma_0$) to prepare $\rho_A$. We propose a new technique called channel block encoding (CBE). We consider applying a quantum channel $\mathcal{C}[\cdot]$ to $\rho_{A_0}$ with the form
\begin{equation}\label{cccnnn}
\mathcal{C}[\rho_{A_0}]=\sum_i \begin{pmatrix}
K_i & 0\\
0& L_i
\end{pmatrix}\rho_{A_0} \begin{pmatrix}
K_i^\dag & 0\\
0& L_i^\dag
\end{pmatrix},
\end{equation}
where $\{K_i\}$ and $\{L_i\}$ are two sets of Kraus operators satisfying $\sum_i K_i^\dag K_i=\sum_i L_i^\dag L_i=I_n$. Focusing on the upper right block $\rho_{A_0}$, we equivalently build the operation in the vectorization picture, resulting in CBE
\begin{eqnarray}
\gamma_0 |A_0\rangle\rightarrow \gamma_0\sum_i K_i\otimes L_i^*  |A_0\rangle.
\end{eqnarray}
\begin{definition}[Channel block encoding (CBE)]
The quantum channel $\mathcal{C}[\cdot]$ is $\eta$-CBE of $O$ if 
\begin{eqnarray}
\sum_i K_i\otimes L_i^*=\eta O.
\end{eqnarray}
\end{definition}
\noindent It can be proved that $\sum_i K_i\otimes L_i^*$ is universal to encode arbitrary operators (see SM \ref{ape5}).

Now for each gate in $V$, we can construct a corresponding CBE and the composite of these channels becomes the CBE $\mathcal{C}_V[\cdot]$ of $V$. Therefore, $\tau$ gates in $V=V_\tau...V_1$ also means $\tau$ channels in $\mathcal{C}_V[\cdot]=\mathcal{C}_{V_\tau}[...\mathcal{C}_{V_1}[\cdot]]$, resulting the same gate complexity. Here, we ignore the complexity inside each channel since a local gate leads to a local CBE channel such that the inside overhead won't change the basic results. If $\mathcal{C}_{V_i}[\cdot]$ is $\eta_i$-CBE of $V_i$, then $\mathcal{C}_V[\cdot]$ is $\eta_V=\eta_\tau...\eta_1$-CBE of $V$. Concrete construction of CBE for elementary gates can be found in SM \ref{ape6}.

\noindent\textit{\textbf{Complexity.—}}
Fig.~\ref{f1}b summarizes the basic framework of the new algorithm. Based on UBSE and DMSE, the value of $\Tr(B^\dag A)$ now can be encoded into forms like $\Tr(U\rho)$ which can be estimated through Hadamard test and amplitude estimation. We give concrete implementation details in SM \ref{ape2} (Fig.~\ref{f2}). The results can be summarized into the following theorem
\begin{theorem}\label{t1a}
Given access to $U_B$ and copies of $\rho_A$, we can estimate $\mu_{r}=\text{Re}[\langle B|A\rangle]$ and $\mu_{i}=\text{Im}[\langle B|A\rangle]$ by doing the Hadamard test (s.q.l.). The query complexity on $U_B$ (including its reverse) and $\rho_A$ to achieve a relative error $\epsilon$ with a success probability at least $1-\delta$ is
\begin{eqnarray}\label{ttts}
T_{q,s}=\mathcal{O}(\gamma^{-2}\lambda^{-2} |\mu_{r(i)}|^{-2}\epsilon^{-2} \log(\delta^{-1})).
\end{eqnarray}
If we are given additional access to $U_{SA}$, the circuit to prepare state $|S_A\rangle$ as the purification $\rho_A$, the amplitude estimation algorithm (h.l.) can be used whose query complexity on $U_B$ and $U_{SA}$ (including there reverse) is
\begin{eqnarray}\label{ttth}
T_{q,h}=\mathcal{O}(\gamma^{-1}\lambda^{-1} |\mu_{r(i)}|^{-1}\epsilon^{-1} \log(\delta^{-1})).
\end{eqnarray}
\end{theorem}
\noindent In the theorem, we use modified versions of amplitude estimation \cite{aaronson2020quantum,grinko2021iterative} avoiding the use of quantum Fourier transform. $U_{SA}$ is the purification circuit of $\mathcal{C}_V[\cdot]$ \cite{preskill1998lecture}. 

Compared with the complexity of queries to the preparation circuits of $|A\rangle$ and $|B\rangle$ in standard approaches, the complexity of queries to $\rho_A$ ($|S_A\rangle$) and $U_B$ in our approach can have $\gamma\lambda$ level of reduction. Importantly, this improvement also holds for the total gate (time) complexity. For $|B\rangle$, both its preparation circuit and the UBSE Pauli operator $U_B$ can be easily prepared. For $|A\rangle$, $V$ and $\mathcal{C}_V[\cdot]$ have the same gate complexity. Thus, the overall improvement of our algorithm directly depends on the value of $\gamma\lambda$.

\noindent\textit{\textbf{Entanglement decides improvements.—}}
We first discuss $U_B$. To have large improvements, we desire a large value of $\lambda$, whose upper bound has the relation with the entanglement of $|B\rangle$ under US-LS bi-partition.
\begin{theorem}[Upper bound on $\lambda$]\label{ubl}
Given a state $|B\rangle$, the value of $\lambda$ in $U_B$ has the upper bound 
\begin{eqnarray}
\lambda\leq \frac{1}{\|B\|_\infty}=2^{\frac{H_\infty(|B\rangle)}{2}},
\end{eqnarray}
where $\|\cdot\|_\infty$ is the spectral norm (largest singular values) and $H_\infty(|B\rangle)=-\log_2(\max_i \beta_i^2)$ is the $\infty$-Rényi entropy \cite{bromiley2004shannon} of $|B\rangle$ under US-LS bi-partition.
\end{theorem}
\noindent This theorem comes from the fact that the spectra norm of unitary operators is 1. $|B\rangle$ with larger entanglement has a larger achievable $\lambda$. When $|B\rangle$ is a maximally entangled state, $B$ will be an un-normalized unitary operator and this unitary operator encodes $B$ with the largest $\lambda=2^{n/2}$. On the contrary,
when $|B\rangle = |\psi_1\rangle|\psi_2\rangle$ is a product state, $\lambda$ has the smallest value $1$. While the upper bound is not necessarily achieved, for the special $|B\rangle$ we considered in this work, $U_B$ is exactly a known Pauli operator with the largest $\lambda=2^{n/2}$. For other $|B\rangle$, see SM \ref{ape8}.

For $\rho_A$, $\gamma$ has an upper bound related to the entanglement of $|A\rangle$.
\begin{theorem}[Upper bound on $\gamma$]\label{ubga}
Given a state $|A\rangle$, the value of $\gamma$ in $\rho_A$ has the upper bound 
\begin{eqnarray}
\gamma \leq \frac{1}{2 \|A\|_{1}}= 2^{-\frac{H_{1/2}(|A\rangle)}{2}-1}, 
\end{eqnarray}
where $\|\cdot\|_1$ is the trace norm (sum of singular values) and $H_{1/2}(|A\rangle):=2\log_2(\sum_i \alpha_i)$ is the $1/2$-Rényi entropy of $|A\rangle$ under US-LS bi-partition.
\end{theorem}
\noindent This theorem comes from the fact that $\rho_A$ is a density matrix (See SM \ref{ape4}). Contrary to $\lambda$, this theorem implies that $|A\rangle$ with larger entanglement has a smaller achievable $\gamma$. For example, when $|A\rangle=|\psi_1\rangle|\psi_2\rangle$ is a product state, we can let $\rho_A=1/2(|0\rangle|\psi_1\rangle+|1\rangle|\psi_2^*\rangle)(\langle 0|\langle\psi_1|+\langle 1|\langle \psi_2^*|)$ to obtain the largest $\gamma=1/2$. When $A=2^{-n/2}I$ corresponds to $|A\rangle$ a maximally entangled state, we can let $\rho_A=2^{-n}|+\rangle\langle +|\otimes I$ such that $\gamma=2^{-n/2-1}$ is the smallest upper bound. In SM \ref{ape4}, we also give the optimal construction of $\rho_A$ saturating the $\gamma$ upper bound.

Since we use $\mathcal{C}_V[\rho_{A_0}]$ to prepare $\rho_A$, we have $\gamma=\eta_V\gamma_0$ which is not necessarily close to the upper bound. Since both $\gamma$ and $\gamma_0$ have upper bounds (Theorem \ref{ubga}), there is also an upper bound for $\eta_v$ by the following corollary.
\begin{corollary}[Upper bound on $\eta_v$ of $V$]\label{coro1}
Given a unitary $V$, the value of $\eta_v$ in CBE of $V$ has the upper bound $\eta_v\leq2^{-\frac{H_{1/2}(V)}{2}}$, where $H_{1/2}(V)$ is defined as $H_{1/2}(V):=\sup_{|\psi\rangle}|H_{1/2}(V|\psi\rangle)-H_{1/2}(|\psi\rangle)|$.
\end{corollary}
\noindent 
Here, $\sup_{|\psi\rangle}$ in $H_{1/2}(V)$ can be chosen allowing ancillary system and initial entanglement \cite{nielsen2003quantum}. Therefore, the upper bound of $\eta_v$ is related to the power of entanglement generation of $V$ under US-LS bi-partition (See SM \ref{ape5}-\ref{ape6}).

\noindent\textit{\textbf{Optimal CBE for gates.—}}
We can construct optimal CBE for quantum circuits in terms of individual single-qubit/two-qubit gates. When a gate $G$ has no interaction between US and LS, its CBE is trivial with $\eta_G=1$, thus, we only need to consider those interaction gates. To achieve the optimal construction for interaction terms, we can use the following channel $\mathcal{C}_G[\cdot]$
\begin{eqnarray}\label{2cbecc}
\mathcal{C}_G[\rho]&&=\frac{s_i}{s_i+s_x+s_y+s_z}\begin{pmatrix}
e^{i\phi_i}I & 0\\
0& I
\end{pmatrix}\rho \begin{pmatrix}
e^{-i\phi_i}I & 0\\
0& I
\end{pmatrix}\nonumber\\&&+\frac{s_x}{s_i+s_x+s_y+s_z}\begin{pmatrix}
e^{i\phi_x}X & 0\\
0& X
\end{pmatrix}\rho \begin{pmatrix}
e^{-i\phi_x}X & 0\\
0& X
\end{pmatrix}\nonumber\\&&+\frac{s_y}{s_i+s_x+s_y+s_z}\begin{pmatrix}
e^{i\phi_y}Y & 0\\
0& -Y
\end{pmatrix}\rho \begin{pmatrix}
e^{-i\phi_y}Y & 0\\
0& -Y
\end{pmatrix}\nonumber\\&&+\frac{s_z}{s_i+s_x+s_y+s_z}\begin{pmatrix}
e^{i\phi_z}Z & 0\\
0& Z
\end{pmatrix}\rho \begin{pmatrix}
e^{-i\phi_z}Z & 0\\
0& Z
\end{pmatrix},
\end{eqnarray}
to encode any 2-qubit unitary operators $G$ connecting US and LS that have the canonical form \cite{tyson2003operator}
\begin{eqnarray}
&&G=e^{i(\theta_x X\otimes X+\theta_y Y\otimes Y+\theta_z Z\otimes Z)}=s_i e^{i\phi_i}I\otimes I+\nonumber\\&&s_x e^{i\phi_x}X\otimes X+s_y e^{i\phi_y}Y\otimes Y+s_z e^{i\phi_z}Z\otimes Z,
\end{eqnarray}
with $\eta_G=1/(s_i+s_x+s_y+s_z)$. Here, we adjust the phases so that $s_i$, $s_x$, $s_y$, and $s_z$ are positive numbers. $\eta_G$ can be proved to achieve the upper bound in Corollary \ref{coro1} (see SM \ref{ape6}). 

Since any two-qubit gates are locally equivalent to the canonical form, their CBE constructions based on Eq. \eqref{2cbecc} are also optimal. We prove the CNOT gate has an optimal CBE efficiency $1/\sqrt{2}$ which means $K$ CNOT gates in $V$ leads to $\eta_V=2^{-K/2}$. Since we can start from $\gamma_0=1/2$ and $\lambda=2^{n/2}$, therefore, we have $\gamma\lambda=2^{(n-K)/2-1}$ and results in Theorem \ref{cbec}. We emphasize that in practice, one can group gates to optimize $\eta_V$. For example, a SWAP gate has its optimal CBE encoding efficiency $1/2$, but if one trivially uses 3 CBE channels of CNOT to build the SWAP gate, the efficiency will be only $2^{-3/2}$. 

\noindent\textit{\textbf{Classical hardness discussion.—}}
The classical hardness can be seen from three aspects (See SM \ref{ooru}-\ref{bbbbbbbb} for more detailed discussions). 

The hardness of $|A\rangle$ inside US and LS: when $|A\rangle=|A_u\rangle|A_u^*\rangle$ and $|B\rangle$ a Bell basis state, we have $\langle B|A\rangle=2^{-n/2}\langle A_u|B|A_u\rangle$. It is known that estimating the Pauli expectation value up to a (polynomial) small additive error is BQP-complete \cite{aharonov2017interactive,janzing2005ergodic}. When $|A\rangle = |A_u\rangle |A_u\rangle$, in Ref. \cite{hangleiter2024bell}, the authors propose the Bell sampling protocol and show that estimating $\langle B|A\rangle$ to a small relative error corresponds to estimating the values of GapP functions to a small relative error, which is GapP-hard \cite{hangleiter2023computational}. For both BQP and GapP amplitudes, our algorithm can give the largest $2^{n/2-1}$ improvement. When the original amplitude is of order $\mathcal{O}(2^{-n/2})$, our algorithm can reduce the overall complexity from exponential to polynomial, which is true for BQP amplitudes. However, for GapP amplitudes, the hardness implies they should be exponentially smaller than $\mathcal{O}(2^{-n/2})$ such that the resulting complexity of our method is still exponential, which coincides with the anti-concentration effect \cite{dalzell2022random,hangleiter2023computational} making almost all amplitudes around $\mathcal{O}(2^{-n})$. Otherwise, quantum computers could solve problems across the entire polynomial hierarchy \cite{toda1991pp}.

The hardness of $|A\rangle$ between US and LS: the CNOT gate restriction across US-LS allows an exponential large bond dimension ($\exp(n)$)), which combining no restrictions on superposition and magic, makes no known efficient classical algorithms \cite{berne1986simulation,gottesman1998heisenberg,cirac2021matrix}. The quantum circuit-cutting methods \cite{cirac2021matrix,peng2020simulating} are also not efficient here indicating that this is a genuine $2n$-qubit regime that cannot be reduced to $n$-qubit quantum systems as in the cases shown above where $|A\rangle$ is a product state. In another view of points, such depth is already enough for the emergence of various interesting things such as the approximate unitary design \cite{schuster2024random}, pseudorandom unitaries \cite{schuster2024random}, and anti-concentration \cite{dalzell2022random}.

The hardness of $|B\rangle$: $|B\rangle$ can go beyond Bell basis states. As long as it is maximally entangled, $B$ is proportional to a unitary operator $U_B$ which can encode universal quantum computing. We can also consider the random behaviors. We hope $|B\rangle$ has a large entanglement. We prove that when $U_B=2^{n/2}B$ is drawn from a unitary 2-design, the ensemble $\{|B\rangle\}$ forms an $\mathcal{O}(2^{-2n})$-approximate $2n$-qubit state 2-design (large entanglement with high probability) in terms of the trace distance in SM \ref{ape9}.

\noindent\textit{\textbf{Additional results.—}}

\textit{Other $|B\rangle$.—}In a general case, $|B\rangle$ can be expressed as the linear combinations of maximally entangled states $|B\rangle=\sum_{i=1}^{2^k} c_i |v_i\rangle$, where $|v_i\rangle$ are maximally entangled. If we have knowledge on this decomposition and the preparation circuit $T_i$ for each $|v_i\rangle=(T_i\otimes I_n)|\Omega\rangle$ with $|\Omega\rangle\propto\sum_i |ii\rangle$. We can use linear combination of unitaries (LCU) \cite{childs2012hamiltonian} (see SM \ref{ape8}) to construct a $n+k$-qubit unitary $U_B$ with $\lambda=2^{n/2}\||B\rangle\|_1^{-1}$ where $\||B\rangle\|_1=\sum_{i=1}^{2^k} |c_i|\leq 2^{k/2}$. When $k=o(n)$, we can expect a large entanglement of $|B\rangle$ and thus an exponential large $\lambda$. Also, the LCU construction of $U_B$ has the same gate complexity as the circuit for preparing $|B\rangle$ \cite{zhang2022quantum}.

\textit{Hamiltonian simulation.—}
$V$ can be also considered as a Hamiltonian simulation unitary. Similar to the circuit case, the result based on the first-order Trotter formula \cite{lloyd1996universal} is summarized in Corollary \ref{cbeh} (also see SM \ref{ape7}).
\begin{corollary}[Hamiltonian simulation by the first-order Trotter formula]\label{cbeh}
Consider the state $|A\rangle=V|A_0\rangle$ with $|A_0\rangle$ a product state and $\rho_{A_0}$ having $\gamma_0=1/2$. Let $V\approx e^{-iHt}$ be the first-order product formula approximation of the dynamics governed by a Hamiltonian $H=H_{US}+ H_{LS}+H_{\text{Inter}}$ where $H_{Inter}$ contains all the interaction terms between US and LS. Then $\rho_A$ can be prepared by applying a CBE channel to $\rho_{A0}$ such that $\gamma$ in $\rho_A$ satisfies $\gamma\approx 2^{-1}e^{-\|H_{\text{Inter}}\|t}\nonumber$, where $\|H_{\text{Inter}}\|$ is the summation of the absolute values of the Pauli coefficients in $H_{\text{Inter}}$.
\end{corollary}
\noindent Therefore, as long as $n/2-\|H_{\text{Inter}}\|t=\Theta(n)$, the improvements is still exponential.

\textit{Gibbs state related.—}
State matrixization framework allows us to probe certain properties of low-temperature Gibbs states using high-temperature ones. Since an $n$-qubit Gibbs state $\rho_\beta=\frac{e^{-\beta H}}{\Tr(e^{-\beta H})}$ is exactly a DMSE of the purified Gibbs state of $\rho_{2\beta}$. To probe properties such as $\langle B|\rho_\beta\rangle$ with $|\rho_\beta\rangle$ denoting the normalized vectorized state from $\rho_\beta$, previously, we need to actually prepare $|\rho_\beta\rangle$ (temperature $1/(2\beta)$) using quantum algorithms \cite{poulin2009sampling,chen2023quantum,zhang2023dissipative}, which can be quantumly hard \cite{lucas2014ising,
kempe2006complexity}, in contrast, using our method, we only need to prepare a Gibbs state at the higher temperature $1/\beta$, easier for quantum computers \cite{bakshi2024high}. In SM \ref{ape10}, we give an example where the query complexity of estimating $\langle B|\rho_\beta\rangle$ is comparable for both methods but with an exponential separation on the preparation complexity between $\rho_\beta$ and $\rho_{2\beta}$.

\noindent\textit{\textbf{Summary and outlook.—}}
In this work, we fundamentally change the way of estimating quantum amplitudes. By encoding the information of the states $|A\rangle$ and $|B\rangle$ into the non-diagonal block of $\rho_A$ and the diagonal block of $U_B$, we can reduce the complexity of estimating $\langle A|B\rangle$. The plausible assumption is that $|A\rangle$ is prepared by a quantum circuit with known construction details and $|B\rangle$ is a known Bell state. $\rho_A$ is prepared by a technique called channel block encoding. The level of improvement is related to the entanglement of the two states: large entanglement of $|B\rangle$ and relatively no large entanglement of $|A\rangle$ leads to even exponential improvements. Our work jumps out of the stereotypes of amplitude estimation that estimating a value $|\mu|$ amplitude requires $\mathcal{O}(|\mu|^{-1})$ time. Our algorithm sheds light on a new mechanism for further quantum speedup and opens up new possibilities for quantum algorithms.

There are several interesting open questions and further directions. First, if we have $\rho_A$ with a $\gamma$ far below the upper bound, can we efficiently amplify it to the upper bound to enjoy our speed-ups? Second, in many situations, we may want to estimate the amplitude of an unknown quantum state with a large entanglement projected on a known product state. Thus, can we switch the roles between $|A\rangle$ and $|B\rangle$? If so, can we efficiently construct $U_B$ with $\lambda$ close to the upper bound? Third, can we rigorously encode classically hard instances with practical interest into amplitudes that fit our algorithm and show more than quadratic quantum speedup? Fourth, we can explore the potential of using channel block encoding to implement non-physical operations such as imaginary time evolution \cite{motta2020determining} and deterministic POVMs \cite{preskill1998lecture} for more applications. 

\noindent\textbf{Email: }\href{shangzx@hku.hk}{Z. S.: shangzx@hku.hk}, \href{zhaoqcs@hku.hk}{Q. Z.: zhaoqcs@hku.hk}. 

\noindent\textbf{Acknowledgments: }
The authors would like to thank Tianfeng Feng, Jue Xu, Zihan Chen, and Siyuan Chen for fruitful discussions; Dominik Hangleiter, Dong An, You Zhou, Naixu Guo, Zhengfeng Ji, Xiaoming Zhang, and Yusen Wu for kindly answering relevant questions; Chaoyang Lu, Xiao Yuan, and Zhenhuan Liu for insightful comments. Z.C.S.\ and Q.Z.\ acknowledge the funding from the HKU Seed Fund for Basic Research for New Staff via Project 2201100596, the Guangdong Natural Science Fund via Project 2023A1515012185, the National Natural Science Foundation of China (NSFC) via Project Nos.\ 12305030 and 12347104, Hong Kong Research Grant Council (RGC) via Project No.\ 27300823, N\_HKU718/23, and R6010-23, Guangdong Provincial Quantum Science Strategic Initiative GDZX2200001.

\bibliography{ref}
\clearpage
\begin{appendix}
\renewcommand\thefigure{\thesection.\arabic{figure}}
\onecolumngrid
\textbf{Supplementary Material of "Entanglement-induced exponential advantage in amplitude estimation via state matrixization"}

Zhong-Xia Shang and Qi Zhao
\renewcommand{\addcontentsline}{\oldacl}
\renewcommand{\tocname}{Contents}
\tableofcontents

\section{Preliminary: Hadamard test and amplitude estimation}\label{ape1}
Given a unitary $U$ and a density matrix $\rho$, Hadamard test \cite{aharonov2006polynomial} can give a standard quantum limit estimate of $\text{Re}[\Tr(U\rho)]$ by measuring the Pauli $Z$ expectation value of an ancilla qubit. Assuming $|\text{Re}[\Tr(U\rho)]|\ll 1 $, we have the following lemma:
\begin{lemma}[Hadamard test \cite{aharonov2006polynomial}]\label{p1}$\\$
To estimate $\mu=\text{Re}[\Tr(U\rho)]$ up to a \textbf{relative} error $\epsilon$ (i.e. $|\bar{\mu}-\mu|\leq \epsilon |\mu|$) with a success probability at least $1-\delta$, the query complexity is:
$$T=\mathcal{O}(|\mu|^{-2}\epsilon^{-2}\log(\delta^{-1})).$$
\end{lemma}
This lemma is derived from Hoeffding's inequality \cite{boucheron2003concentration} and is tight whenever $|\text{Re}[\Tr(U\rho)]|\ll 1 $ corresponds to the case where the variance is comparable with the bound of the random variable ($\langle Z\rangle$).

Given two state $|S_1\rangle$ and $|S_2\rangle$, amplitude estimation \cite{brassard2002quantum,aaronson2020quantum,grinko2021iterative} can give a Heisenberg limit estimate of $|\langle S_2|S_1\rangle|$ given access on the Grover operator \cite{grover1996fast}:
\begin{eqnarray}
U_G=-(I-2|S_1\rangle\langle S_1|)(I-2|S_2\rangle\langle S_2|)
.\end{eqnarray}
We have the following lemma:
\begin{lemma}[Amplitude estimation \cite{aaronson2020quantum,grinko2021iterative}]\label{p2}$\\$
To estimate $\mu=|\langle S_2|S_1\rangle|$ up to an \textbf{additive} error $\epsilon$ (i.e. $|\bar{\mu}-\mu|\leq \epsilon$) with a success probability at least $1-\delta$, the query complexity is:
$$T=\mathcal{O}(\epsilon^{-1}\log(\delta^{-1})).$$
\end{lemma}
Note that Lemma \ref{p2} adopts the results from Ref. \cite{aaronson2020quantum,grinko2021iterative} where there is no need to build quantum Fourier transform circuit as shown in the original amplitude estimation protocol \cite{brassard2002quantum}.

\section{Overview \label{overview}}
The algorithm proposed in this work contains a paradigm shift. While $\langle B|A\rangle$ is determined by $|A\rangle$ and $|B\rangle$, these states are merely assemblies of complex numbers. We propose using alternative quantum objects to encode the same information as $|A\rangle$ and $|B\rangle$ to jump out of the limitations of amplitude estimation. 

The basic framework has been summarized in Fig.~\ref{f1}. Suppose that we have $2n$-qubit states $|A\rangle$ and $|B\rangle$ with the expressions $|A\rangle=\sum_{ij}\alpha_{ij}|i\rangle|j\rangle$ and $|B\rangle=\sum_{ij}\beta_{ij}|i\rangle|j\rangle$ (When the number of qubits is odd, we can always add an ancilla qubit to meet the case.), we first turn them into $n$-qubit matrices $A=\sum_{ij}\alpha_{ij}|i\rangle\langle j|$ and $B=\sum_{ij}\beta_{ij}|i\rangle\langle j|$ which we call the matrixization procedure. For example, under matrixization, a 2-qubit Bell state $(|0\rangle|0\rangle+|1\rangle|1\rangle)/\sqrt{2}$ will be turned into an un-normalized single-qubit identity operator $I/\sqrt{2}$. We will also call the inverse procedure the vectorization procedure. Their formal definitions are given below:
\begin{definition}[Vectorization and Matrixization]
Given a matrix $O=\sum_{ij}o_{ij}|i\rangle\langle j|$, the vectorization mapping $\mathcal{V}$ is defined as follows: $$\mathcal{V}[O]:=\ket{\ket{O}}=\sum_{ij}o_{ij}|i\rangle|j\rangle.$$
Matrixization is the reverse mapping of vectorization $\mathcal{M}[\ket{\ket{O}}]=O$.
\end{definition}
\noindent We can also understand these two mappings from Choi–Jamiołkowski isomorphism \cite{choi1975completely}. 
In the following, we will call the system of $|i\rangle$ in $\ket{\ket{O}}$ with $i$ the row index of $O$ as the upper subsystem (US) and call the system of $|j\rangle$ with $j$ the column index as the lower subsystem (LS). 

When $|A\rangle$ and $|B\rangle$ are turned into matrices, we have $\langle B|A\rangle=\Tr(B^\dag A)$. Thus, the current task is to find quantum objects that encode (A) and (B) in such a way that the value of $\Tr(B^\dagger A)$ can be estimated. In quantum computing, we are familiar with the expression of the form $\Tr(U\rho)$, which the Hadamard test or the amplitude estimation can evaluate. The similarity between $\Tr(U\rho)$ and $\Tr(B^\dag A)$ suggests encoding $B$ into a unitary operator $U_B$ and $A$ into a density matrix $\rho_A$. In this way, we can estimate the amplitude under the matrixization picture with the help of the Hadamard test or the amplitude estimation. The benefit of matrixization will be clear in the later sections. 

We want to encode the $2n$-qubit pure state $|B\rangle$ and $|A\rangle$ into a $2^n\times 2^n$ diagonal block of a 
larger unitary operator $U_B$ and density matrix $\rho_A$,
referred to as unitary block state encoding (UBSE) and density matrix state encoding (DMSE),  respectively. The formal definitions of these encoding methods are as follows. 
\begin{definition}[Unitary block state encoding (UBSE)]
Given an $n+k$-qubit unitary operator $U_B$, if $U_B$ satisfies:
$$\mathcal{V}[(I_n\otimes \langle 0|_k ) U_B (I_n\otimes |0\rangle_k)]=\lambda |B\rangle,$$
where $|0\rangle_k$ is the all-zero computational basis of the $k$-qubit ancilla system, $\lambda\geq 0$, and $|B\rangle$ is a $2n$-qubit pure state, then $U_B$ is called an $(n+k,\lambda)$-UBSE of $|B\rangle$.
\end{definition}
\noindent In other words, $U_B$ is a block encoding \cite{low2019hamiltonian} of $B=\mathcal{M}[|B\rangle]$. Meanwhile, DMSE can defined as:
\begin{definition}[Density matrix state encoding (DMSE)]
Given an $l+n$-qubit density matrix $\rho_A$, if $\rho_A$ satisfies:
$$\mathcal{V}[(\langle s_1|_l\otimes I_n) \rho_A (|s_2\rangle_l\otimes I_n)]=\gamma |A\rangle,$$
where $|s_1\rangle_l$ and $|s_2\rangle_l$  are two computational basis states of the $l$-qubit ancilla system satisfying $\langle s_1|s_2\rangle_l=0$, $\gamma\geq 0$, and $|A\rangle$ is a $2n$-qubit pure state, then $\rho_A$ is called an $(l+n,s_1,s_2,\gamma)$-NDSE of $|A\rangle$.
\end{definition}
\noindent The reason we use the non-diagonal block of density matrices is to get rid of the Hermitian and positive semi-definite restrictions of density matrices. In this work, we will mainly focus on $(1+n,0,1,\gamma)$-DMSE where $\rho_A$ has the form:
\begin{eqnarray}
\rho_A=\begin{pmatrix}
\cdot & \gamma A \\
\gamma A^\dag & \cdot
\end{pmatrix}.
\end{eqnarray}
Note that $\rho_A$ is also a $(1+n,1,0,\gamma)$-DMSE of $|A^\dag\rangle$.

\section{Estimating amplitudes by DMSE and UBSE (Theorem \ref{t1a})}\label{ape2}
Given the two quantum states $|A\rangle$ and $|B\rangle$, in this work, we are interested in estimating $\langle B |A\rangle$, the amplitude of the state $|A\rangle$ on the basis state $|B\rangle$. Throughout this work, we set $|A\rangle$ and $|B\rangle$ as $2n$-qubit states. We define the first $n$-qubit subsystem as the upper subsystem (US) and the last $n$-qubit subsystem as the lower subsystem (LS). Then, under the US-LS bi-partition, $|A\rangle$ and $|B\rangle$ have the Schmidt decomposition:
\begin{eqnarray}
|A\rangle=\sum_i \alpha_i |a_{ui}\rangle|a_{li}\rangle\text{, }
|B\rangle=\sum_i \beta_i |b_{ui}\rangle|b_{li}\rangle
,\end{eqnarray}
where $\{\alpha_i\}$ and $\{\beta_i\}$ are real and positive and we have $\{|a_{ui}\rangle\}$, $\{|a_{li}\rangle\}$, $\{|b_{ui}\rangle\}$, and $\{|b_{li}\rangle\}$ satisfy $\langle a_{ui}|a_{uj}\rangle=\langle a_{li}|a_{lj}\rangle=\langle b_{ui}|b_{uj}\rangle=\langle b_{li}|b_{lj}\rangle=\delta_{ij}$. 

We now show how to use DMSE and UBSE to estimate $\langle B|A\rangle$. We talk about the query complexity of estimating $\langle B|A\rangle$ in terms of $U_B$ and the preparation of $\rho_A$ (and also their conjugate transposes). We leave the constructions of $U_B$ and $\rho_A$ in the following section. We will consider two models that will lead to a standard quantum limit estimation and a Heisenberg limit estimation respectively. The concrete circuit implementation can be found in Fig. \ref{f2}.

\begin{figure}[ht]
\centering
\subfigure[]{
\Qcircuit @C=2.4em @R=1.6em {
\lstick{}  & \ctrlo{3} & \ctrl{3} & \qw &\qw  & \ctrl{1}&\qw \\
\lstick{}  & \qw & \qw & \ctrlo{1} & \ctrl{1} & \ctrl{0}&\qw \\
\lstick{}  & \qw & \qw &  \gate{X} & \gate{iY}&\qw&\qw\\
\lstick{}  & \multigate{1}{U_B} & \multigate{1}{U_B^\dag}  &\qw & \qw &\qw & \qw  &\dstick{U_B}\\
\lstick{}  & \ghost{U_B} & \ghost{U_B^\dag} & \qw & \qw  &\qw&\qw\\
\protect\inputgroupv{3}{4}{.8em}{.8em}{\rho_A}\\
\protect\gategroup{4}{7}{5}{7}{.8em}{\}}
}}

\subfigure[]{
\Qcircuit @C=2.4em @R=1.6em {
\lstick{} & \gate{H} & \multigate{4}{W_r} & \gate{H} & \qw \\
\lstick{} & \gate{H} & \ghost{W_r} & \gate{H} & \qw \\
\lstick{} & \qw & \ghost{W_r} & \qw & \qw \\
\lstick{} & \qw & \ghost{W_r} & \qw & \qw &\dstick{U_B}\\
\lstick{} & \qw & \ghost{W_r} & \qw & \qw \\
\protect\inputgroupv{3}{4}{.8em}{.8em}{\rho_A}\\
\protect\gategroup{4}{5}{5}{5}{.8em}{\}}
}}
\hspace{20mm}
\subfigure[]{
\Qcircuit @C=2.4em @R=1.6em {
\lstick{} & \gate{H} & \ctrl{1} & \gate{H} & \qw \\
\lstick{} & \qw & \multigate{1}{W_r} & \qw & \qw \\
\lstick{} & \qw & \ghost{W_r}\qwx[1]& \qw & \qw \\
\lstick{} & \qw &  \qw\qwx[1] & \qw & \qw \\
\lstick{} & \qw & \multigate{2}{W_r}& \qw & \qw \\
\lstick{} & \qw & \ghost{W_r} & \qw & \qw &\dstick{U_B}\\
\lstick{} & \qw & \ghost{W_r} & \qw & \qw \\
\protect\inputgroupv{4}{6}{.8em}{.8em}{|S_A\rangle}\\
\protect\gategroup{6}{5}{7}{5}{.8em}{\}}
}}
\caption{(a): The circuit diagram of $W_r$. For $W_i$, we can simply turn the $X$, $iY$, and CZ gates in $W_r$ into $iX$, $Y$, and $
|00\rangle\langle00|-|01\rangle\langle01|-|10\rangle\langle10|-|11\rangle\langle11|$. The systems are arranged in the order $\{1,1,1,n,k\}$. (b): The circuit diagram of $U_{B,r1}$. For $U_{B,i1}$, we can simply turn $W_r$ to $W_i$. The systems are arranged in the order $\{1,1,1,n,k\}$. (c): The circuit diagram of $U_{B,r2}$. For $U_{B,i2}$, we can simply turn $W_r$ to $W_i$. The systems are arranged in the order $\{1,1,1,m,1,n,k\}$.\label{f2}}
\end{figure}

\subsection{Standard quantum limit}
The first model is given access on $\rho_A$ ($(1+n,0,1,\gamma)$-DMSE) and $U_B$. First, we need to combine $U_B$ with linear combinations of unitaries (LCU) \cite{childs2012hamiltonian} to build two $2+1+n+k$-qubit \footnote{In this work, expressions like $2+1+n+k$-qubit means the whole system is arranged with the order $(2,1,n,k)$ where each symbol denotes a subsystem which should be clear in the text with the corresponding number of qubits.} unitaries $U_{B,r1}$ and $U_{B,i1}$. Their definitions and functions follow the lemma below:
\begin{lemma}[$U_{B,r1}$, $U_{B,i1}$]\label{l1}$\\$

$U_{B,r1}=(H\otimes H\otimes I_{1+n+k}) W_r (H\otimes H\otimes I_{1+n+k})$ with unitary: $$W_r=|00\rangle\langle 00|\otimes X \otimes U_B+i|01\rangle\langle 01|\otimes Y \otimes U_B+|10\rangle\langle 10|\otimes X \otimes U_B^\dag-i|11\rangle\langle 11|\otimes Y \otimes U_B^\dag,$$ satisfies $$(\langle 00|\otimes I_{1+n}\otimes \langle 0|_k)U_{B,r1}(| 00\rangle\otimes I_{1+n}\otimes |0\rangle_k)=\frac{1}{2}\begin{pmatrix}
0 & \lambda B\\
\lambda B^\dag & 0
\end{pmatrix}.$$

$U_{B,i1}=(H\otimes H\otimes I_{1+n+k}) W_i (H\otimes H\otimes I_{1+n+k})$ with unitary: $$W_i=i|00\rangle\langle 00|\otimes X \otimes U_B-|01\rangle\langle 01|\otimes Y \otimes U_B-i|10\rangle\langle 10|\otimes X \otimes U_B^\dag-|11\rangle\langle 11|\otimes Y \otimes U_B^\dag,$$ satisfies $$(\langle 00|\otimes I_{1+n}\otimes \langle 0|_k)U_{B,i1}(| 00\rangle\otimes I_{1+n}\otimes |0\rangle_k)=\frac{1}{2}\begin{pmatrix}
0 & i\lambda B\\
-i\lambda B^\dag & 0
\end{pmatrix}.$$
\end{lemma}

Based on the Lemma. \eqref{l1}, we have the following relations:
\begin{eqnarray}
\frac{\Tr(U_{B,r1}|00\rangle\langle 00|\otimes \rho_A\otimes |0\rangle\langle 0|_k)}{\gamma\lambda}&&= \text{Re}[\langle B|A\rangle],\label{g1}\\
\frac{\Tr(U_{B,i1}|00\rangle\langle 00|\otimes \rho_A\otimes |0\rangle\langle 0|_k)}{\gamma\lambda}&&= \text{Im}[\langle B|A\rangle].\label{g2}
\end{eqnarray}
Thus, the LHS of Eq. \eqref{g1}-\eqref{g2} can be used to estimate $\langle B |A\rangle$ from $\rho_A$ and $U_B$ and their forms indicate we can use Hadamard tests \cite{aharonov2006polynomial} to evaluate their values. According to Lemma \ref{p1}, the query complexity of estimating $\mu_r=\text{Re}[\langle B|A\rangle]$ to a relative error $\epsilon$ with a success probability at least $1-\delta$ using Eq \eqref{g1} is:
\begin{eqnarray}\label{c1}
T_{s1}=\mathcal{O}(\gamma^{-2}\lambda^{-2} |\mu_r|^{-2}\epsilon^{-2} \log(\delta^{-1})).
\end{eqnarray}

In contrast, if we directly measure $\mu_r$ by Hadamard test in terms of $|A\rangle$ and $|B\rangle$ rather than $\rho_A$ and $U_B$, the complexity is:
\begin{eqnarray}\label{c2}
T_{s2}=\mathcal{O}( |\mu_r|^{-2}\epsilon^{-2} \log(\delta^{-1})).
\end{eqnarray}
Thus, there is a $\gamma^{-2}\lambda^{-2}$ complexity improvement. The imaginary part has a similar conclusion.

Note that, in this work, we mainly focus on exponentially small amplitudes, thus $T_{s2}$ is tight. However, we will see, even $\text{Re}[\langle B|A\rangle]$ is very small, $\Tr(U_{B,r1}|00\rangle\langle 00|\otimes \rho_A\otimes |0\rangle\langle 0|_k)$ can be of order $\mathcal{O}(1)$, which means the true complexity of our method can be significantly small than $T_{s1}$ and Bernstein inequality \cite{boucheron2003concentration} should be applied instead of Hoeffding's inequality used in Lemma \ref{p1}. Nonetheless, this loose complexity \eqref{c1} is already enough to show the advantages of our method and thus, we will use it throughout this work.

\subsection{Heisenberg limit}
The second model is given access to $U_B$ and a $m+1+n$-qubit pure state $|S_A\rangle$ which is the purification of $\rho_A$ i.e. $Tr_m(|S_A\rangle\langle S_A|)=\rho_A$ with $Tr_m(\cdot)$ the partial trace on the $m$-qubit ancilla system. We need to combine $U_B$ with LCU to build two $1+2+m+1+n+k$-qubit unitaries $U_{B,r2}$ and $U_{B,i2}$ inspired by the construction in Ref. \cite{rall2020quantum}. Their definitions and functions follow the lemma below:
\begin{lemma}[$U_{B,r2}$, $U_{B,i2}$]\label{l2}$\\$

$U_{B,r2}=( H\otimes I_{2+m+1+n+k}) V_r ( H\otimes I_{2+m+1+n+k})$ with unitary: $$V_r=|0\rangle\langle 0|\otimes I_{2+m+1+n+k}+|1\rangle\langle 1|\otimes I_m \otimes W_r,$$ satisfies $$(\langle 000|\otimes I_m\otimes I_{1+n}\otimes \langle 0|_k)U_{B,r2}(| 000\rangle\otimes  I_m\otimes I_{1+n}\otimes |0\rangle_k)=\frac{1}{4}I_m\otimes\begin{pmatrix}
2I & \lambda B\\
\lambda B^\dag &2I
\end{pmatrix}.$$

$U_{B,i2}=( H\otimes I_{2+m+1+n+k}) V_i ( H\otimes I_{2+m+1+n+k})$ with unitary: $$V_i=|0\rangle\langle 0|\otimes I_{2+m+1+n+k}+|1\rangle\langle 1|\otimes I_m \otimes W_i,$$ satisfies $$(\langle 000|\otimes I_m\otimes I_{1+n}\otimes \langle 0|_k)U_{B,i2}(| 000\rangle\otimes  I_m\otimes I_{1+n}\otimes |0\rangle_k)=\frac{1}{4}I_m\otimes\begin{pmatrix}
2I & i\lambda B\\
-i\lambda B^\dag &2I
\end{pmatrix}.$$
\end{lemma}
where in the lemma descriptions, the term $I_m \otimes W_{r(i)}$ has an interchange between $m$ and $2$ to simplify the expression. 

Based on the Lemma. \eqref{l2}, we have the following relations:
\begin{eqnarray}
\frac{2|\langle 000| \langle S_A| \langle 0|_k U_{B,r2} |0\rangle_k |S_A\rangle |000\rangle|-1}{\gamma\lambda}&&=\text{Re}[\langle B|A\rangle],\label{g3}\\
\frac{2|\langle 000| \langle S_A| \langle 0|_k U_{B,i2} |0\rangle_k |S_A\rangle |000\rangle|-1}{\gamma\lambda}&&=\text{Im}[\langle B|A\rangle].\label{g4}
\end{eqnarray}
\begin{proof}
First, we have:
\begin{eqnarray}
\langle 000|\langle  S_A| \langle 0|_k U_{B,r2} |0\rangle_k |S_A\rangle |000\rangle&&=\frac{1}{4}\Tr\left(|S_A\rangle\langle S_A| \left(I_m\otimes \begin{pmatrix}
2I & \lambda B\nonumber\\
\lambda B^\dag &2I
\end{pmatrix} \right)\right)\nonumber\\&&=\frac{1}{4}\Tr\left(\rho_A \begin{pmatrix}
2I & \lambda B\\
\lambda B^\dag &2I
\end{pmatrix}\right)\nonumber\\&&=\frac{1}{2}+\frac{\gamma\lambda}{4}(\Tr(AB^\dag)+\Tr(A^\dag B))=\frac{1}{2}+\frac{\gamma\lambda}{4}(\langle B|A\rangle+\langle A|B\rangle)\nonumber\\&&=\frac{1}{2}+\frac{\gamma\lambda Re[\langle B|A\rangle]}{2}\nonumber
.\end{eqnarray}
Since $\|\lambda B\|_{\infty}\leq 1$ with $\|\cdot\|_{\infty}$ denoting the spectral norm (largest singular value), $\begin{pmatrix}
2I & \lambda B\\
\lambda B^\dag &2I
\end{pmatrix}$ is positive semi-definite. Thus, we get:
$$|\langle 000|\langle  S_A| \langle 0|_k U_{B,r2} |0\rangle_k |S_A\rangle |000\rangle|=\frac{1}{2}+\frac{\gamma\lambda Re[\langle B|A\rangle]}{2}.$$
The proof of the imaginary part is similar.
\end{proof}

Thus, the LHS of Eq. \eqref{g3}-\eqref{g4} can be used to estimate $\langle B |A\rangle$ from $|S_A\rangle$ and $U_B$ and their forms indicate we can use amplitude estimation \cite{aaronson2020quantum,grinko2021iterative} to evaluate their values. To do the amplitude estimation to estimate $\mu_r=\text{Re}[\langle B|A\rangle]$, we need to build the Grover operator \cite{grover1996fast}: 
\begin{eqnarray}
U_{G,r}&&=-(I-2U_{B,r2} (|0\rangle\langle 0|_k\otimes |S_A\rangle\langle S_A|\otimes |000\rangle\langle 000|)U_{B,r2}^\dag)\nonumber\\&&(I-2|0\rangle\langle 0|_k\otimes |S_A\rangle\langle S_A|\otimes |000\rangle\langle 000|)\nonumber\\&&=-(I-2U_{B,r2}( |0\rangle\langle 0|_k \otimes( U_{SA}|0\rangle\langle 0|_{m+1+n}U_{SA}^\dag)\otimes |000\rangle\langle 000|)U_{B,r2}^\dag))\nonumber\\&&(I-2|0\rangle\langle 0|_k \otimes (U_{SA}|0\rangle\langle 0|_{m+1+n}U_{SA}^\dag)\otimes |000\rangle\langle 000|)
,\end{eqnarray}
where $U_{SA}$ is the unitary that prepares $|S_A\rangle$ from $|0\rangle_{m+1+n}$. Thus, the Grover operator can be built from $U_B$, $U_B^\dag$, $U_{SA}$, and $U_{SA}^\dag$. Then, according to Lemma \ref{p2}, the query complexity of estimating $\mu_r=\text{Re}[\langle B|A\rangle]$ to a relative error $\epsilon$ which corresponds to an additive error $\gamma\lambda|\mu_r|\epsilon/2$ in the amplitude estimation with a success probability at least $1-\delta$ using Eq \eqref{g1} is: 
\begin{eqnarray}\label{c3}
T_{h1}=\mathcal{O}(\gamma^{-1}\lambda^{-1} |\mu_r|^{-1}\epsilon^{-1} \log(\delta^{-1}))
.\end{eqnarray}
In contrast, if we directly measure $\mu_r$ by amplitude estimation in terms of $|A\rangle$ and $|B\rangle$ with the aid of the construction in Ref. \cite{rall2020quantum}, the complexity is:
\begin{eqnarray}\label{c4}
T_{h2}=\mathcal{O}( |\mu_r|^{-1}\epsilon^{-1} \log(\delta^{-1})).
\end{eqnarray}
Thus, there is a $\gamma^{-1}\lambda^{-1}$ complexity improvement. The imaginary part has a similar conclusion.

\subsection{Estimating $\gamma$ and $\lambda$}\label{ape3}
While the values of $\lambda$ and $\gamma$ can be directly concluded from the way of building $U_B$ and preparing $\rho_A$, there may be cases where we have no prior knowledge of the exact values of $\gamma$ in $\rho_A$ (also $|S_A\rangle$) and maybe also $\lambda$ in $U_B$. Thus, we need measurement strategies to evaluate these values in advance to estimate amplitudes. We give discussions here.

To give the formal analysis, we need the following error propagation lemma:
\begin{lemma}[Error propagation]\label{lll2}
Given the function $\frac{X}{YZ}$ of three independent random variables $X$, $Y$, and $Z$, to estimate $\frac{X}{YZ}$ to a relative error $\epsilon$, it is sufficient to estimate $X$, $Y$, and $Z$ to a relative error $\frac{\epsilon}{\sqrt{3}}$.
\end{lemma}
\begin{proof}
Define $E(X)=\mu_x$, $E(Y)=\mu_y$, and $E(Z)=\mu_z$. First, we have:
$$Var\left(\frac{X}{YZ}\right)\approx\frac{\mu_y^2\mu_z^2 Var(X)+\mu_x^2Var(YZ)}{\mu_y^4 \mu_z^4},$$
and we have:
$$Var(YZ)\approx \mu_y^2 Var(Z)+\mu_z^2 Var(Y)+Var(Y)Var(Z).$$
These relations come from the Taylor series approximation of estimators \cite{kroese2013handbook}. Thus, we get:
$$Var\left(\frac{X}{YZ}\right)=\frac{\mu_y^2\mu_z^2 Var(X)+\mu_x^2(\mu_y^2 Var(Z)+\mu_z^2 Var(Y)+Var(Y)Var(Z)))}{\mu_y^4 \mu_z^4}.$$
We can ask $ Var(X)=\epsilon^2\mu_x^2/3$, $ Var(Y)=\epsilon^2\mu_y^2/3$, and $ Var(Z)=\epsilon^2\mu_z^2/3$ with $\epsilon\ll 1$, then we have:
$$Var\left(\frac{X}{YZ}\right)=\frac{\mu_y^2\mu_z^2 Var(X)+\mu_x^2(\mu_y Var(Z)+\mu_z Var(Y)+Var(X)Var(Y)))}{\mu_y^4 \mu_z^4}\approx \frac{\epsilon^2\mu_x^2}{ \mu_y^2 \mu_z^2}.$$
Thus, the lemma is proved. Note that, in this proof, there is no need to consider bias for small enough relative errors since the variance takes the main portion and is quadratically larger than the bias.
\end{proof}

Based on this lemma, and the form of Eq. \eqref{g1}, Eq. \eqref{g2}, Eq. \eqref{g3}, and Eq. \eqref{g4}, to estimate the amplitude $\langle B|A\rangle$ to a relative error $\epsilon$, it is sufficient to estimate $\gamma$ and $\lambda$ to a relative error $\epsilon$ where we omit constant factors.

To estimate $\gamma$ in $\rho_A$ ($(1+n,0,1,\gamma)$-DMSE), we have the following relation:
\begin{eqnarray}\label{gamma}
\gamma^2=\frac{\Tr(((|0 1\rangle\langle 10|+|10\rangle\langle 01| )\otimes\text{SWAP}_{n+n})\rho_A\otimes\rho_A)}{2}.
\end{eqnarray}
Note that in the expression, to enable concise presentations, the index order of the operator side is $(1,1,n,n)$ while the index order of the density matrix side is $(1,n,1,n)$.
\begin{proof}
\begin{eqnarray}
&&\Tr(((|0 1\rangle\langle 10|+|10\rangle\langle 01| )\otimes\text{SWAP}_{n+n})\rho_A\otimes\rho_A)\nonumber\\&&=\Tr((|0 1\rangle\langle 10| \otimes\text{SWAP}_{n+n})(|1 0\rangle\langle 01|\otimes \gamma^2(A^\dag\otimes A)))\nonumber\\&&+\Tr((|1 0\rangle\langle 01| \otimes\text{SWAP}_{n+n})(|01 \rangle\langle 10|\otimes \gamma^2(A\otimes A^\dag)))\nonumber\\&&=\gamma^2 \Tr(\text{SWAP}_{n+n}(A^\dag\otimes A))+\gamma^2 \Tr(\text{SWAP}_{n+n}(A\otimes A^\dag))\nonumber\\&&=2\gamma^2\langle A|A\rangle=2\gamma^2\nonumber.
\end{eqnarray}
\end{proof}
Since $\gamma\geq0$, the RHS of Eq. \eqref{gamma} can be used to estimate its value. Here, we can simply use the operator averaging method \cite{mcclean2016theory} to do the estimation where we need to apply a rotation unitary to rotate $\rho_A\otimes\rho_A$ to the diagonal basis of $((|0 1\rangle\langle 10|+|10\rangle\langle 01| )\otimes\text{SWAP}_{n+n})$ and do repeated measurements for the evaluation. Since $((|0 1\rangle\langle 10|+|10\rangle\langle 01| )\otimes\text{SWAP}_{n+n})$ has a spectrum within the range $[-1,1]$ and a $\epsilon$-relative estimation of $\gamma$ corresponds to a $4\gamma^2\epsilon $-additive estimation of $\Tr(((|0 1\rangle\langle 10|+|10\rangle\langle 01| )\otimes\text{SWAP}_{n+n})\rho_A\otimes\rho_A)$ according to the error propagation rule \cite{kroese2013handbook}, the sampling complexity to achieve the purpose is of order $\mathcal{O}(\gamma^{-4}\epsilon^{-2})$. The above procedures can be easily generalized to $|S_A\rangle$ with similar conclusions. Again, one can use advanced measurement techniques like the amplitude estimation here to reduce the complexity to the Heisenberg limit.

To estimate $\lambda$ in $U_B$, we have the following relation:
\begin{eqnarray}\label{lambda}
\lambda^2=2^{2n}\Tr\left(((U_B\otimes U_B^\dag) (\text{SWAP}_{n+n} \otimes I_k\otimes I_k)) \rho_{I0}\otimes \rho_{I0}\right),
\end{eqnarray}
with $\rho_{I0}=\frac{1}{2^n}I_n\otimes|0\rangle\langle 0|_k$. Note that in the expression, to enable concise presentations, the index order of the operator side is $(n,n,k,k)$ while the index order of the density matrix side is $(n,k,n,k)$.
\begin{proof}
\begin{eqnarray}
&&\Tr\left(((U_B\otimes U_B^\dag) (\text{SWAP}_{n+n} \otimes I_k\otimes I_k)) \left(\left(\frac{1}{2^n}I_n\right)\otimes \left(\frac{1}{2^n}I_n\right)\otimes |0\rangle\langle 0|_k\otimes |0\rangle\langle 0|_k\right)\right)\nonumber\\=&&\Tr\left(\lambda^2 (B\otimes B^\dag\text{SWAP}_{n+n}\left(\frac{1}{2^{2n}}I_{n+n}\right)\right)=\frac{\lambda^2}{2^{2n}}\Tr(B^\dag B)=\frac{\lambda^2}{2^{2n}}\langle B|B\rangle=\frac{\lambda^2}{2^{2n}}\nonumber.
\end{eqnarray}
\end{proof}
Since $\lambda\geq0$, the RHS of Eq. \eqref{lambda} can be used to estimate its value. Since $((U_B\otimes U_B^\dag) (\text{SWAP}_{n+n} \otimes I_k\otimes I_k))$ is unitary, we can again use the Hadamard test to do the estimation. A $\epsilon$-relative estimation of $\lambda$ corresponds to a $2^{-(2n-1)}\lambda^2\epsilon $-additive estimation of $\Tr\left(((U_B\otimes U_B^\dag) (\text{SWAP}_{n+n} \otimes I_k\otimes I_k)) \rho_{I0}\otimes \rho_{I0}\right)$ according to the error propagation rule, the sampling complexity to achieve the purpose is of order $\mathcal{O}(2^{4n}\lambda^{-4}\epsilon^{-2})$. Again, one can use amplitude estimation to further reduce the complexity. Note that, in most cases that we are considering in this work, $\lambda=\mathcal{O}(2^{n/2})$, thus, the sampling complexity is inevitably exponential. 

\section{About $U_B$}
In our setting, the state $|B\rangle$ serves as a basis state on which we want to project the state $|A\rangle$. Thus, we require $|B\rangle$ is a known state and show basic constructions of $U_B$.

\subsection{Fundamental limit}
Since $|B\rangle$ has the Schmidt form: $|B\rangle=\sum_i \beta_i |b_{ui}\rangle|b_{li}\rangle$, it is easy to see that $\lambda$ in $U_B$ has the upper bound in Theorem \ref{ubl} due to the restrictions of block encoding \cite{lin2022lecture}. This theorem connects the upper bound of $\lambda$ with the entanglement Rényi entropy of $|B\rangle$ between US and LS: $|B\rangle$
with large entanglement has large achievable $\lambda$ and $|B\rangle$ with small entanglement has large achievable $\lambda$. For example,
when $|B\rangle = |\psi_1\rangle|\psi_2\rangle$ is a product state, $\lambda$ has the smallst value $1$. On the other hand, when $|B\rangle$ is a product
of Bell states (i.e. 1st and $(n + 1)$th qubits form a Bell state, 2nd and $(n + 2)$th qubits form a Bell state...) which corresponds to the case where all $\beta_i$ are equal, it has the
largest entanglement, $\lambda$ can have the largest value $2^{n/2}$.

\subsection{Bell basis \label{bbss}}
A particularly interesting class of basis states are the Bell basis states which are the states considered in the main text. 
\begin{definition}[Bell basis]
A $2n$-qubit Bell basis state has the form:
$|B_b\rangle=|b_1\rangle|b_2\rangle...|b_n\rangle$
where $|b_i\rangle$ is one of 2-qubit Bell states acting on the $i$th (US) and $(n+i)$th (LS) qubits.
\end{definition}
Since for the two-qubit case, we have the following correspondence:
\begin{eqnarray}
|\Phi^+\rangle&&=\frac{1}{\sqrt{2}}(|0\rangle|0\rangle+|1\rangle|1\rangle)\xrightarrow[]{\mathcal{M}}\frac{1}{\sqrt{2}}I\nonumber,\\
|\Phi^-\rangle&&=\frac{1}{\sqrt{2}}(|0\rangle|0\rangle-|1\rangle|1\rangle)\xrightarrow[]{\mathcal{M}}\frac{1}{\sqrt{2}}Z\nonumber,\\
|\Psi^+\rangle&&=\frac{1}{\sqrt{2}}(|0\rangle|1\rangle+|1\rangle|0\rangle)\xrightarrow[]{\mathcal{M}}\frac{1}{\sqrt{2}}X\nonumber,\\
|\Psi^-\rangle&&=\frac{1}{\sqrt{2}}(|0\rangle|1\rangle-|1\rangle|0\rangle)\xrightarrow[]{\mathcal{M}}\frac{i}{\sqrt{2}}Y.
\end{eqnarray}
Thus, each Bell basis state $|B_b\rangle$ is mapped to a $n$-qubit Pauli operator. Since Pauli operators are unitary, they are directly the UBSE of $|B_b\rangle$ with the largest $\lambda=2^{n/2}$.

\subsection{$U_B$ by unitary decomposition}\label{ape8}
From Theorem \ref{ubl}, it is easy to find that the Pauli operators can be generalized to an arbitrary $n$-qubit unitary operator $V$ which is a UBSE of a $2n$-qubit maximally entangled state $|v\rangle=2^{-n/2}\mathcal{V}[V]$ with the largest $\lambda=2^{n/2}$. To have advantages on the amplitude estimation by our method, we hope $\lambda$ to be as large as possible. Thus, it is natural to set $|B\rangle$ as a small number of linear combinations of maximally entangled states:
\begin{eqnarray}\label{bform}
|B\rangle=\sum_{i=1}^{2^k} c_i |v_i\rangle,
\end{eqnarray}
and we know this decomposition priorly.

Now, the task is to construct $U_B$, which can be done by the LCU method \cite{low2019hamiltonian} and leads to the following theorem:
\begin{lemma}[$U_B$ by LCU\label{cccb}]
Define a $k$-qubit unitary $T$ satisfies:
$$T|0\rangle_k= \frac{1}{\sqrt{\sum_j |c_j|}}\sum_{i=1}^{2^k} \sqrt{|c_i|}|i\rangle_k,$$
and a $n+k$-qubit unitary $Q$:
$$Q=\sum_{i=1}^{2^k}  \frac{c_i}{|c_i|}V_i\otimes |i\rangle\langle i|_d ,$$
with $V_i=2^{n/2}\mathcal{M}[|v_i\rangle]$.
Then the $n+k$-qubit unitary $U_B=(I_n\otimes T^\dag) Q(I_n\otimes T)$ is a $(n + k, \lambda)$-UBSE of $|B\rangle$ with $\lambda=2^{n/2}\||B\rangle\|_1^{-1}\geq 2^{(n-k)/2}$ where $\||B\rangle\|_1=\sum_{i=1}^{2^k} |c_i|\leq 2^{k/2}$ denotes the vector 1-norm of $|B\rangle$.
\end{lemma}
\begin{proof}
\begin{eqnarray}
(I_n\otimes \langle 0|_k)(I_n\otimes T^\dag) Q(I_n\otimes T)(I_n\otimes | 0\rangle_k)&&=\frac{1}{\||B\rangle\|_1}(I_n\otimes \sum_{i=1}^{2^k} \sqrt{|c_i|} \langle i|_k)Q(I_n\otimes \sum_{j=1}^{2^k} \sqrt{|c_j|} | j\rangle_k)\nonumber\\&&=\frac{\sum_{i=1}^{2^k} c_i V_i}{\||B\rangle\|_1}=\frac{2^{n/2}}{\||B\rangle\|_1}B\nonumber.
\end{eqnarray}
\end{proof}

It is worth noting that given the description of Eq. \ref{bform}, the gate cost of constructing $U_B$ equals the optimal gate cost of preparing $|B\rangle$ \cite{zhang2022quantum}. 
Here, we only consider an LCU-based UBSE. It would be interesting to see if other block encoding methods \cite{low2019hamiltonian,camps2024explicit} will lead to more efficient UBSE protocols and get benefits from other access models beyond Eq. \eqref{bform}. It is possible to use the methods in Ref. \cite{guo2021nonlinear,rattew2023non} to construct $U_B$ for any $|B\rangle$, however, this can only give $\lambda=1$ leading to no improvements to the amplitude estimation.

\section{About $\rho_A$ ($|S_A\rangle$)}
Given the state $|A\rangle$, it is crucial to ask how to prepare the state $\rho_A$ ($|S_A\rangle$) and how efficient the preparation can be. From the previous section, we have seen that to make our method have advantages over traditional direct amplitude measurement methods, we hope to make $\gamma$ as large as possible. Not only this can reduce the complexity $T_{s1}$ and $T_{h1}$, but also this can reduce the complexity of estimating $\gamma$ when it is unknown. At the same time, the complexity of preparing $\rho_A$ ($|S_A\rangle$) is also important since this is hidden in $T_{s1}$ and $T_{h1}$ and a high such complexity can cancel out the advantage of our method. In this section, we give a thorough discussion on various aspects on constructing $\rho_A$ ($|S_A\rangle$).

\subsection{Fundamental limit}\label{ape4}
Suppose $A$ has the singular value decomposition $A=W_A^\dag\Sigma_A V_A$ where $W_A$ and $V_A$ are unitary and $\Sigma_A$ is a diagonal, Hermitian, and positive semi-definite matrix, it is then easy to check that the diagonal elements of $\Sigma_A$ are exactly Schmidt coefficients $\{\alpha_i\}$ of $|A\rangle$: $|A\rangle=\sum_i \alpha_i |a_{ui}\rangle|a_{li}\rangle$ with $\langle a_{ui}|a_{uj}\rangle=\langle a_{li}|a_{lj}\rangle=\delta_{ij}$. Following this, we obtain the theorem below:
\begin{theorem}[A detailed statement of Theorem \ref{ubga}]\label{ubgamma}
Given a state $|A\rangle$, the value of $\gamma$ in $\rho_A$ ($|S_A\rangle$) has the upper bound:
$$\gamma \leq \frac{1}{2 \|A\|_{1}}= 2^{-\frac{H_{1/2}(|A\rangle)}{2}-1},$$
where $\|\cdot\|_1$ is the trace norm (sum of singular values) and $H_{1/2}(|A\rangle):=2\log_2(\sum_i \alpha_i)$ is the $1/2$-Rényi entropy of $|A\rangle$ under the partition between US and LS. $\rho_A$ achieving this bound has the form:
\begin{eqnarray}
\rho_A=\begin{pmatrix}
\gamma W_A^\dag  \Sigma_A W_A & \gamma A \\
\gamma A^\dag & \gamma V_A  \Sigma_A V_A^\dag
\end{pmatrix}.
\end{eqnarray}
\end{theorem}
\begin{proof}
We can do a basis transformation to $\rho_A$:
\begin{eqnarray}
\begin{pmatrix}
W_A& 0 \\
0 & V_A
\end{pmatrix}\rho_A\begin{pmatrix}
W_A^\dag& 0 \\
0 & V_A^\dag
\end{pmatrix}=\begin{pmatrix}
\cdot & \gamma W_A A V_A^\dag \\
\gamma V_A A^\dag W_A^\dag & \cdot
\end{pmatrix}=\begin{pmatrix}
\cdot & \gamma \Sigma_A \\
\gamma \Sigma_A & \cdot
\end{pmatrix}\nonumber
.\end{eqnarray}
Under this basis, we should have $\gamma \alpha_i\leq \sqrt{p_{0i}p_{1i}}$ with $p_{0i}$ the diagonal probability that shares the same row with $\gamma \alpha_i$ in the upper right block and $p_{1i}$ the diagonal probability that shares the same column with $\gamma \alpha_i$ in the upper right block. The reason for this is that this 2-dimensional subspace $\{0i, 1i\}$ must be an un-normalized density matrix due to the definition of the density matrix. Thus, we have:
\begin{eqnarray}
&&\sum_i\gamma \alpha_i \leq \sum_i\sqrt{p_{0i}p_{1i}}\leq \sum_i\frac{p_{0i}+p_{1i}}{2}=\frac{1}{2}\nonumber\\&&\rightarrow \gamma \leq \frac{1}{2 \|A\|_{1}}=2^{-\frac{H_{1/2}(|A\rangle)}{2}-1}\nonumber
.\end{eqnarray}
If $\gamma$ reaches the bound, then $\begin{pmatrix}
\gamma \Sigma_A & \gamma \Sigma_A \\
\gamma \Sigma_A & \gamma \Sigma_A
\end{pmatrix}$ is a legal density matrix. Thus, $\rho_A$ achieving the bound can be obtained:
\begin{eqnarray}\label{ubcons}
\rho_A=\begin{pmatrix}
W_A^\dag & 0 \\
0 & V_A^\dag
\end{pmatrix} \begin{pmatrix}
\gamma \Sigma_A & \gamma \Sigma_A \\
\gamma \Sigma_A & \gamma \Sigma_A
\end{pmatrix} \begin{pmatrix}
W_A & 0 \\
0 & V_A
\end{pmatrix}=\begin{pmatrix}
\gamma W_A^\dag \Sigma_A W_A & \gamma A \\
\gamma A^\dag & \gamma V_A  \Sigma_A V_A^\dag
\end{pmatrix}\nonumber
.\end{eqnarray}
\end{proof}
This theorem connects the upper bound of $\gamma$ with the entanglement Rényi entropy of $|A\rangle$ between US and LS: $|A\rangle$ with large entanglement has small achievable $\gamma$ and $|A\rangle$ with small entanglement has large achievable $\gamma$. For example, when $|A\rangle=|\psi_1\rangle|\psi_2\rangle$ is a product state, $\gamma$ can have the largest value $1/2$. On the other hand, when $|A\rangle$ is a product of Bell states (i.e. 1st and $(n + 1)$th qubits form a Bell state, 2nd and $(n + 2)$th qubits form a Bell state...), it has the largest entanglement, $\gamma$ can not be larger than $2^{-n/2-1}$.

When we have prior knowledge of the Schmidt decomposition of $|A\rangle$, we can then use the construction procedure Eq. \eqref{ubcons} to prepare $\rho_A$ achieving the upper bound. However, in general cases, $|A\rangle$ is unknown and is prepared by a unitary operator from an initial state, thus, we need tools to construct quantum operations in the vectorization picture.

\subsection{Channel block encoding}\label{ape5}
Here, we give a general construction of $\rho_A$ by a new technique which we call the channel block encoding (CBE). Here and after, we will only talk about constructions of $\rho_A$, since for any quantum channel, we can always use the Stinespring dilation \cite{stinespring1955positive} to contain the channel in a unitary operator ($U_{SA}$) of an enlarged system (The system of $|S_A\rangle$).

We consider applying a quantum channel $\mathcal{C}[\cdot]$ to an initial $n+1$-qubit density matrix $\rho_{ini}$ with the form:
\begin{equation}\label{qiuqiu}
\mathcal{C}[\rho_{ini}]=\sum_i \begin{pmatrix}
K_i & 0\\
0& L_i
\end{pmatrix}\rho_{ini} \begin{pmatrix}
K_i^\dag & 0\\
0& L_i^\dag
\end{pmatrix},
\end{equation}
where $\{K_i\}$ and $\{L_i\}$ are two sets of Kraus operators satisfying $\sum_i K_i^\dag K_i=\sum_i L_i^\dag L_i=I_n$. Under the vectorization, we have:
\begin{equation}
\mathcal{V}[\mathcal{C}[\rho_{ini}]]=C\ket{\ket{\rho_{ini}}}=\sum_i \begin{pmatrix}
K_i\otimes K_i^* & 0 &0 &0\\
0& K_i\otimes L_i^* & 0 &0 \\
0& 0& L_i\otimes K_i^* & 0 \\
 0 &0 &0& L_i\otimes L_i^*  \\
\end{pmatrix}\begin{pmatrix}
\ket{\ket{\rho_{ini,00}}} \\\ket{\ket{\rho_{ini,01} }}\\\ket{\ket{\rho_{ini,10}}}\\\ket{\ket{\rho_{ini,11}}} \\
\end{pmatrix},
\end{equation}
where $C$ is the matrix representation of the channel $\mathcal{C}[\cdot]$ in the vectorization picture. Focusing on the upper right block $\rho_{ini,01}$, we have:
\begin{eqnarray}
\ket{\ket{\rho_{ini,01} }}\rightarrow \sum_i K_i\otimes L_i^* \ket{\ket{\rho_{ini,01} }}.
\end{eqnarray}
Thus, the operation $\sum_i K_i\otimes L_i^*$ is the action we can do in the vectorization picture. 

In fact, $\sum_i K_i\otimes L_i^*$ is universal in the sense that we add express any operators in this form up to a normalization factor. To see this, we can consider a channel with a more restricted form: 
\begin{eqnarray}\label{ucbe}
\mathcal{C}[\rho_{ini}]=\sum_i p_i\begin{pmatrix}
W_i & 0\\
0& V_i
\end{pmatrix}\rho_{ini} \begin{pmatrix}
W_i^\dag & 0\\
0& V_i^\dag
\end{pmatrix},
\end{eqnarray}
with $\{W_i,V_i\}$ unitary operators and $\sum_i p_i=1$. We will call such restricted forms the unitary channel block encoding (UCBE). Under the vectorization,  $\rho_{ini,01}$ is transformed to:
\begin{eqnarray}\label{uucbe}
\ket{\ket{\rho_{ini,01} }}\rightarrow \sum_i p_i W_i\otimes V_i^* \ket{\ket{\rho_{ini,01} }}.
\end{eqnarray}
The action $\sum_i p_i W_i\otimes V_i^*$ is then a natural linear combination of unitaries. Since any $2n$-qubit Pauli operator is a tensor product of two $n$-qubit Pauli operators, thus, the form of the action $\sum_i p_i W_i\otimes V_i^*$ means that, for any $2n$-qubit operator, we can always express it as a linear combination of $2n$-qubit Pauli operators and encode it into a $n+1$-qubit quantum channel. We want to mention here that the condition $\{p_i\}$ are positive is not a restriction since for any term with the form like $|p_i|e^{i\theta_i} W_i\otimes V_i^*$, we can always let the unitaries absorb the phase. For example, given an $2n$-qubit operator $O$ with the form:
\begin{equation}\label{ode}
O=\sum_i g_i P_i\otimes Q_i,
\end{equation}
with $g_i=|g_i|e^{i\theta_i}$ and $\{P_i\}$ and $\{Q_i\}$ $n$-qubit Pauli operators, we can build a $n+1$-qubit quantum channel $\mathcal{C}_o[\cdot]$ of the form Eq. \eqref{ucbe} and set $W_i=e^{i\theta_i} P_i$, $V_i^*=Q_i$, and $p_i=|g_i|/\sum_i |g_i|$, then the channel is a CBE of $O$ satisfies:
\begin{equation}\label{ocbe}
(\langle 01|\otimes I)C_o(|01\rangle\otimes I)=\frac{O}{\|O\|},
\end{equation}
with $\|O\|:=\sum_i |g_i|$. We want to emphasize here that while Pauli basis CBE can encode arbitrary operators, it might be far from optimal for some operators. The formal definition of CBE is summarized. 
\begin{definition}[Channel block encoding (CBE)]\label{cbe}
Given a $2n$-qubit operator $O$, if we can find an $l+n$-qubit quantum channel of the form Eq. \eqref{qiuqiu} whose matrix form $C$ satisfies the condition:
$$\|\eta O-(\langle s|\otimes I)C(|s\rangle\otimes I)\|\leq\varepsilon,$$
with the encoding efficiency $\eta\geq0$, a binary string $s$ and an $l$-qubit computational basis $|s\rangle$, then the channel is an $(l+n, s,\eta,\varepsilon)$-CBE of $O$.
\end{definition}
\noindent For example, Eq. \eqref{ocbe} indicates that the channel $\mathcal{C}_o[\cdot]$ is a $(1+n,01,\sum_i |g_i|,0)$-CBE of $O$. 

Having the tools of CBE, we can therefore build various actions to quantum states (in the vectorization picture). In most cases, $|A\rangle$ is prepared by a quantum process such as a quantum circuit or a Hamiltonian simulation. The basic workflow to convert $|A\rangle$ to $\rho_A$ is then try to use CBE to rebuild the quantum process in the vectorization picture to get $\rho_A$ which we will give detailed discussions in the following. 

Now, we discuss using CBE to construct unitaries. Starting from an initial state $\gamma_0 |A_0\rangle$ encoded in $\rho_{A_0}$, if we have $|A\rangle=V|A_0\rangle$ and the unitary $V$ is realized by a CBE with $\eta_v$, then $\rho_A$ encodes $|A\rangle$ with $\gamma=\eta_v\gamma_0$. Since both $\gamma$ and $\gamma_0$ have upper bounds decided by $|A\rangle$ and $|A_0\rangle$, it is thus simple to give an upper bound for $\eta_v$ of $V$, which gives the Corollary \ref{coro1}. Therefore, the upper bound of $\eta_v$ is connected with the power of entanglement generation of $V$. Since $\gamma=\eta_v\gamma_0$, if the initial state is a product state ($\gamma_0=1/2$), this indicates that we require the entanglement power of $V$ should be relatively small to make $\gamma$ large. The reason that $H_{1/2}(V)$ is defined allowing local ancilla qubits and initial entanglement of $|\psi\rangle$ since these resources can enlarge the entangling power compared with the restricted setting with only product initial state with no ancillas \cite{leifer2003optimal,nielsen2003quantum}. Note that this corollary only restricts the entanglement power of $V$ between US and LS but does not restrict its power inside US and LS.  

While Corollary \ref{coro1} tells us the best efficiency we can expect when using CBE to encode a unitary, it doesn't tell us if it is possible to find and how to find a channel reaching this bound. In fact, given any operator $V$ ($O$), there can be various channels that are CBE of $V$ ($O$) with different efficiencies. Again, since $\gamma=\eta_v\gamma_0$, it is important to find constructions with high efficiencies close to the upper bound. For general unitaries, finding optimal constructions is highly complicated, and no efficient methods are known. In the following, we will discuss strategies for using CBE to build quantum circuits and Hamiltonian simulations and give optimal CBE constructions for individual gates and Trotter terms.

\subsection{CBE for circuit (Theorem \ref{cbec})}\label{ape6}
First, we formally summarize three measures of the entangling power of unitaries including $H_{1/2}(V)$: 
\begin{definition}[Entangling power of unitary operators]$\\$\label{ep}
The first definition of the entangling power of $V$ is $H_{1/2}(V)$ defined as the maximum entanglement generations in terms of $1/2$-Rényi entropy over all possible initial states allowing ancilla qubits and initial entanglement):
$$H_{1/2}(V):=\sup_{|\psi\rangle}|H_{1/2}(V|\psi\rangle)-H_{1/2}(|\psi\rangle)|.$$

The second definition of the entangling power of $V$ is $H_{1/2,p}(V)$ defined as the maximum entanglement generations in terms of $1/2$-Rényi entropy over all possible product initial states (between US and LS) allowing ancilla qubits:
$$H_{1/2,p}(V):=\sup_{|\psi\rangle=|\psi_u\rangle|\psi_l\rangle}|H_{1/2}(V|\psi\rangle)-H_{1/2}(|\psi\rangle)|.$$

The third definition of the entangling power of $V$ is $H_{1/2,s}(V)$ defined as the $1/2$-Rényi entropy in terms of the operator Schmidt decomposition of $V$:
$$H_{1/2,s}(V):=2\log_2(\sum_i \frac{s_i}{\sqrt{\sum_i s_i^2}})=2\log_2(\sum_i s_i)-2n,$$
where the operator Schmidt decomposition \cite{nielsen2003quantum} $V=\sum_i s_i X_i\otimes Y_i$ satisfies $s_i\geq 0$ and $\Tr(X_i^\dag X_j)=\Tr(Y_i^\dag Y_j)=\delta_{ij}$. The second equality is because $\Tr(V^\dag V)=\sum_i s_i^2=2^{2n}$.
\end{definition}
The relations between these three measures are summarized by the following lemma:
\begin{lemma}[Relations between measures]\label{epl}
$$H_{1/2}(V)\geq H_{1/2,p}(V)\geq H_{1/2,s}(V).$$
\end{lemma}
\begin{proof}
The relation $H_{1/2}(V)\geq H_{1/2,p}(V)$ is obvious due to the definition of these two measures. For $H_{1/2,p}(V)\geq H_{1/2,s}(V)$, it is sufficient to prove there exists a product state $|\psi\rangle=|\psi_u\rangle|\psi_l\rangle$ such that $|H_{1/2}(V|\psi\rangle)-H_{1/2}(|\psi\rangle)|=H_{1/2,s}(V)$ \cite{nielsen2003quantum}. When $|\psi_u\rangle=2^{-n/2}\sum_k|k\rangle|k\rangle$ ($|\psi_l\rangle$) itself is a maximally entangled state between the $n$-qubit half system and $n$-qubit ancilla system in US (LS), we have:
\begin{eqnarray}\label{babala}
V|\psi\rangle=\sum_i s_i X_i\otimes I_n |\psi_u\rangle\otimes Y_i\otimes I_n |\psi_l\rangle\nonumber
.\end{eqnarray}
Since we have:
\begin{eqnarray}
\langle \psi_u |(X_j^\dag\otimes I_n) (X_i\otimes I_n) |\psi_u\rangle&&=\frac{1}{2^n}\Tr((X_j^\dag\otimes I_n) (X_i\otimes I_n)\sum_{kl}|k\rangle\langle l|\otimes |k\rangle\langle l|)\nonumber\\
&&=\frac{1}{2^n}\Tr(X_j^\dag X_i I_n)=2^{-n}\delta_{ji}\nonumber
,\end{eqnarray}
a similar relation also holds for $\{Y_i\otimes I_n |\psi_l\rangle\}$.
Thus, Eq. \eqref{babala} is exactly the Schmidt decomposition of $V|\psi\rangle$ and we get:
\begin{eqnarray}
|H_{1/2}(V|\psi\rangle)-H_{1/2}(|\psi\rangle)|=2\log_2(2^{-n}\sum_i  s_i)=H_{1/2,s}(V)\nonumber.
\end{eqnarray}
which finishes the proof.
\end{proof}

Now, we show how to use CBE to build quantum circuits. Given a quantum circuit composed of single-qubit and 2-qubit gates, we can divide these gates into two types. The first type contains gates acting only on US or LS. For such a gate, we can build it in the vectorization picture by trivially implementing a UCBE of the form Eq. \eqref{ucbe} with $p_1=1$ (i.e. not a channel but a unitary) to achieve an efficiency $\eta=1$. 

The second type contains 2-qubit gates that give interactions between US and LS. For such a gate $U_G$, we can first write it in the form of canonical decomposition \cite{tyson2003operator}:
\begin{eqnarray}
U_G=G_{u1}\otimes G_{l1}e^{i(\theta_x X\otimes X+\theta_y Y\otimes Y+\theta_z Z\otimes Z)}G_{u2}\otimes G_{l2}
,\end{eqnarray}
where $G_{u1}$, $G_{l1}$, $G_{u2}$, and $G_{l2}$ are single-qubit unitaries. Since both $G_{u1}\otimes G_{l1}$ and $G_{u2}\otimes G_{l2}$ belong to the first type, their CBE realizations are trivial, thus, we only need to consider the CBE construction of $G=e^{i(\theta_x X\otimes X+\theta_y Y\otimes Y+\theta_z Z\otimes Z)}$ which has the (un-normalized) operator Schmidt decomposition:
\begin{eqnarray}
G=s_i e^{i\phi_i}I\otimes I+s_x e^{i\phi_x}X\otimes X+s_y e^{i\phi_y}Y\otimes Y+s_z e^{i\phi_z}Z\otimes Z.
\end{eqnarray}
It is easy to see that a channel $\mathcal{C}_G$ with the form of UCBE Eq. \eqref{ucbe} can do the CBE of $G$:
\begin{eqnarray}\label{2cbe}
\mathcal{C}_G[\rho]&&=\frac{s_i}{s_i+s_x+s_y+s_z}\begin{pmatrix}
e^{i\phi_i}I & 0\\
0& I
\end{pmatrix}\rho \begin{pmatrix}
e^{-i\phi_i}I & 0\\
0& I
\end{pmatrix}\nonumber\\&&+\frac{s_x}{s_i+s_x+s_y+s_z}\begin{pmatrix}
e^{i\phi_x}X & 0\\
0& X
\end{pmatrix}\rho \begin{pmatrix}
e^{-i\phi_x}X & 0\\
0& X
\end{pmatrix}\nonumber\\&&+\frac{s_y}{s_i+s_x+s_y+s_z}\begin{pmatrix}
e^{i\phi_y}Y & 0\\
0& -Y
\end{pmatrix}\rho \begin{pmatrix}
e^{-i\phi_y}Y & 0\\
0& -Y
\end{pmatrix}\nonumber\\&&+\frac{s_z}{s_i+s_x+s_y+s_z}\begin{pmatrix}
e^{i\phi_z}Z & 0\\
0& Z
\end{pmatrix}\rho \begin{pmatrix}
e^{-i\phi_z}Z & 0\\
0& Z
\end{pmatrix},
\end{eqnarray}
with an efficiency $\eta_G=(s_i+s_x+s_y+s_z)^{-1}$. For $\mathcal{C}_G$, we have the following theorem:
\begin{theorem}[CBE of 2-qubt gates]\label{cbe2}
$\mathcal{C}_G$ is an optimal construction of $G$.
\end{theorem}
\begin{proof}
To show the optimality, we need to show $\eta_G$ reaches the upper bound in Corollary \ref{coro1}. First, we have:
$$\eta_G=2^{-\log_2(s_i+s_x+s_y+s_z)}=2^{\frac{-H_{1/2,s}(G)}{2}}.$$
Since $H_{1/2,s}(G)\leq H_{1/2}(G)$, thus, we have:
$$\eta_G\geq 2^{\frac{-H_{1/2}(G)}{2}}.$$
However, due to the Corollary \ref{coro1}, we also have:
$$\eta_G\leq 2^{\frac{-H_{1/2}(G)}{2}}.$$
Thus, we must have:
$$\eta_G= 2^{\frac{-H_{1/2}(G)}{2}}.$$
\end{proof}
This proof also indicates the following corollary:
\begin{corollary}[Equivalence of measures of entangling power for two-qubit gates]
When $V$ is a two-qubit gate, we have:
$$H_{1/2}(V)= H_{1/2,p}(V)= H_{1/2,s}(V).$$
\end{corollary}
It would be interesting to see if this corollary works for general $2n$-qubit systems.

Therefore, for any 2-qubit gate belongs to the second type, we have shown how to do its optimal CBE. For example, the CNOT gate (with Schmidt number 2) is locally equivalent with:
\begin{eqnarray}
G_{cnot}=\frac{1}{\sqrt{2}}I\otimes I+\frac{i}{\sqrt{2}}X\otimes X.
\end{eqnarray}
Thus, its CBE efficiency is $\eta_c=1/\sqrt{2}$. And the SWAP gate (with Schmidt number 4) is locally equivalent with:
\begin{eqnarray}
G_{swap}&&=(2^{-3/2}+i2^{-3/2})I\otimes I+(2^{-3/2}+i2^{-3/2})X\otimes X\nonumber\\&&+(2^{-3/2}+i2^{-3/2})Y\otimes Y+(2^{-3/2}+i2^{-3/2})Z\otimes Z.
\end{eqnarray}
Thus, its CBE efficiency is $\eta_s=1/2$. Now, for general quantum circuits, we can repeatedly apply channels for CBE of gates to obtain the Theorem \ref{cbec}. Here, we want to mention that for multi-qubit gates, it is non-trivial to find optimal constructions since there is no known generalization of the canonical decomposition for large systems.

\subsection{CBE for Hamiltonian simulation (Corollary \ref{cbeh})}\label{ape7}
We now consider the CBE constructions of Hamiltonian simulations. Here, for simplicity, we consider Hamiltonians with known forms under Pauli basis: $H=\sum_i h_i P_i$ with $\{P_i\}$ $2n$-qubit Pauli operators and we only consider the first-order product formula (Trotter-Suzuki formula) formalism \cite{lloyd1996universal}. To simulate the unitary $e^{-iH t}$, the first-order product formula has the form:
\begin{eqnarray}\label{pf}
e^{-iHt}\approx U_{pf}(t,r)=(\Pi_i e^{-i h_i P_it/r})^r
.\end{eqnarray}
To use $U_{pf}$ to simulate $e^{-iHt}$ to an accuracy $\varepsilon$, the gate complexity is of order $\mathcal{O}(t^2\varepsilon^{-1})$. 

Similar to the circuit case, to do the CBE for $U_{pf}(t,r)$, we can divide Pauli terms in $H$ in two types. The first type contains Pauli terms with identity $I_n$ either on US or LS. For these terms, we can trivially implement a UCBE of the form Eq. \eqref{ucbe} with $p_1=1$ (i.e. not a channel but a unitary) to achieve an efficiency $\eta=1$ for $e^{-i h_i P_i t/r}$. 

The second type contains Pauli terms that have $\{X,Y,Z\}$ on both US and LS. For these terms, $e^{-i h_i P_i t/r}$ can build interactions between US and LS and thus we need a channel to do the CBE. To do so, we can observe that each $P_i$ (we will take $P_0$ as an example) can be transformed into $X\otimes I_{n-1}\otimes X\otimes I_{n-1}$ by local Clifford circuits:
\begin{eqnarray}
P_0=U_{cu}\otimes U_{cl}(X\otimes I_{n-1}\otimes X\otimes I_{n-1}) U_{cu}^\dag\otimes U_{cl}^\dag
,\end{eqnarray}
with $U_{cu}$ and $U_{cl}$ $n$-qubit Clifford circuits acting on US and LS respectively. Thus, it is sufficient to consider the CBE of $e^{-ih_0 X\otimes X t/r}$ which have been introduced in Eq. \eqref{2cbe}. Thus, for $e^{-i h_0 P_0 t/r}$, Eq. \eqref{2cbe} gives its optimal CBE construction with an efficiency $\eta_{P0}=1/(\cos(|h_0|t/r)+\sin(|h_0|t/r))$. With this optimal construction, we have the Theorem \ref{cbeh} for CBE of $U_{pf}(t,r)$. The proof is shown below:
\begin{proof}[Proof of Theorem \ref{cbeh}]
We can define the efficiency of CBE of $U_{pf}(t,r)$ as $\eta_H$, then we have:
\begin{eqnarray}
\eta_H=\Pi_{i\in \text{Type 2}}(\cos(|h_i|t/r)+\sin(|h_i|t/r))^{-r}\nonumber.
\end{eqnarray}
Assuming $r$ is relatively large such that $\cos(|h_i|t/r)+\sin(|h_i|t/r)\approx 1+|h_i|t/r$, thus we have:
\begin{eqnarray}
\eta_H&&\approx \Pi_{i\in \text{Type 2}}(1+|h_i|t/r)^{-r}\approx (1+\sum_{i\in \text{Type 2}}|h_i|t/r)^{-r}\nonumber\\&& =(1+\|H\|_{\text{Inter}}t/r)^{-r}=\left((1+\|H\|_{\text{Inter}}t/r)^{-r/(\|H\|_{\text{Inter}}t)}\right)^{\|H\|_{\text{Inter}}t}\approx e^{-\|H\|_{\text{Inter}}t}.
\end{eqnarray}
\end{proof}

In Theorem \ref{cbec}, we use the number of CNOT gates as a benchmark, as it is a standard gate in many universal gate sets. In Theorem \ref{cbeh}, we consider only the first-order Trotter formula \cite{lloyd1996universal} for simplicity. Nevertheless, the result in Theorem \ref{cbeh} should be nearly tight and cannot be optimized by high-order Trotter \cite{berry2007efficient} and other simulation methods such as LCU \cite{berry2015simulating} and QSP \cite{low2017optimal} as it matches the direct implementation of $e^{-iHt}$ when we assume a sufficiently large Trotter steps for the proof. We emphasize that in practice, one can use various strategies to optimize the encoding efficiency in Theorem \ref{cbec} and \ref{cbeh}. For example, a SWAP gate has its optimal CBE encoding efficiency $1/2$, but if one trivially uses 3 CBE channels of CNOT to build the SWAP gate, the efficiency will be only $2^{-3/2}$.

\section{Workflow of the algorithm}

\begin{figure}[htbp]
\centering
\includegraphics[width=0.5\textwidth]{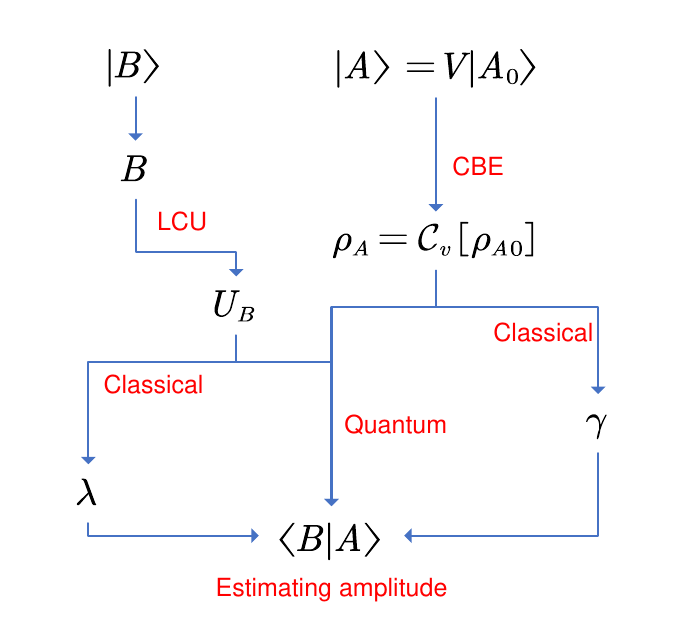}
\caption{Workflow of the amplitude estimation algorithm. To estimate the amplitude $\langle B|A\rangle$, we need to construct $U_B$ and prepare $\rho_A$. For $U_B$, we first turn $|B\rangle$ of the form Eq. \eqref{bform} into its matrix form and then use LCU to do the block encoding of $B$ and obtain $U_B$. Since the form of $|B\rangle$ is known priorly, the value of $\lambda=2^{n/2}\||B\rangle\|_1^{-1}$ can be directly calculated from the LCU construction. For $\rho_A$, according to $|A\rangle=V|A_0\rangle$, we first prepare $\rho_{A_0}$ and construct a channel $\mathcal{C}_v[\cdot]$ as the CBE of $V$, then we can obtain $\rho_A=\mathcal{C}_v[\rho_{A_0}]$. Since we can easily calculate the CBE efficiency of $V$ from its construction, the value of $\gamma$ can also be classically calculated. Now, we can use these elements with the Hadamard test to give a s.q.l. estimation of the amplitude $\langle B|A\rangle$ or with the amplitude estimation to give a h.l. estimation of the amplitude $\langle B|A\rangle$ based on Fig. \ref{f2}. Note that for h.l. estimation, we need to replace $\rho_A$ with its purification state $|S_A\rangle$.\label{f3}}
\end{figure}

Now, we can summarize the whole workflow of our algorithm in Fig.~\ref{f3}. Note that based on the above discussions, the exact values of $\lambda$ and $\gamma$ can be directly concluded from the construction of $U_B$ and the preparation of $\rho_A$. 

\section{Improvements on estimating amplitudes\label{improve}}
The interesting point of our method is that we compress $2n$-qubit information in $n$-qubit systems, but this can surprisingly give benefits: we can estimate the amplitude $\langle B|A\rangle$ with complexity reductions $\gamma^{-2}\lambda^{-2}$ and $\gamma^{-1}\lambda^{-1}$ in Theorem \ref{t1a} compared with direct estimations by Hadamard tests and amplitude estimations. The improvement is more significant with a larger $\gamma\lambda$. 
In an extreme case where $|A\rangle$ is a product state between US and LS and $|B\rangle$ is a maximally entangled state between US and LS, we are allowed to have the largest value $\gamma\lambda=2^{n/2-1}$ with $\gamma=1/2$ and $\lambda=2^{n/2}$, which is an exponential improvement over traditional methods. In this case, we summarize the query complexity costs to estimate the amplitude with a relative error $\epsilon=\mathcal{O}(1)$ in Table \ref{tbl1}. 
\begin{table}[htbp]
\centering  
\begin{tabular}{|c|c|c|c|c|}  
\hline  
& & & & \\
Amplitude & Direct estimation (s.q.l.)&Direct estimation (h.l.)&Our method (s.q.l.)&Our method (h.l.) \\  
& & & & \\
\hline
& & & &\\ 
$\mathcal{O}(2^{-n/2})$&$\mathcal{O}(2^{n})$&$\mathcal{O}(2^{n/2})$&$\mathcal{O}(1)$&$\mathcal{O}(1)$ \\
& & & & \\
\hline
& & & &\\ 
$\mathcal{O}(2^{-n})$&$\mathcal{O}(2^{2n})$&$\mathcal{O}(2^{n})$&$\mathcal{O}( 2^{n})$&$\mathcal{O}( 2^{n/2})$ \\
& & & & \\
\hline
\end{tabular}
\caption{Comparisons of query complexities for an $\epsilon=\mathcal{O}(1)$-relative estimation of $\langle B|A\rangle$. We consider the extreme cases where $|A\rangle$ is a product state between US and LS and has been encoded into $\rho_A$ achieving the upper bound in Theorem \ref{ubga}, and $|B\rangle$ is a maximally entangled state between US and LS and has been encoded into $U_B$ achieving the upper bound in Theorem \ref{ubl}. We consider two scenarios: $\text{Re(Im)}[\langle B|A\rangle]$ is of order $\mathcal{O}(2^{-n/2})$ and is of order $\mathcal{O}(2^{-n})$, which are related with discussions in the implication section.\label{tbl1}}  
\end{table}

Thus, we can see that when $\text{Re(Im)}[\langle B|A\rangle]$ is of order $\mathcal{O}(2^{-n/2})$, we can achieve an exponential speedup over traditional methods. Since we are considering $2n$-qubit systems, this corresponds to situations where $|A\rangle$ is relatively concentrated on $|B\rangle$. When $\text{Re(Im)}[\langle B|A\rangle]$ is of order $\mathcal{O}(2^{-n})$ which is the cases where $|A\rangle$ is anti-concentrated on $|B\rangle$, we can achieve a complexity using only a standard quantum limit estimation (Eq.~\eqref{g1}-\eqref{g2}) comparable with the direct estimation (h.l.), which also means we can achieve an equivalent quadratic speedup over the direct estimation (s.q.l.) without changing the dependence on $\epsilon$. Moreover, in this example, by using the Heisenberg limit estimation (Eq.~\eqref{g3}-\eqref{g4}), we can even achieve an equivalent quadratic speedup over the direct estimation (h.l.) and an equivalent quartic speedup over the direct estimation (s.q.l.). Note that we are not actually achieving for example the truly quartic speedup with respect to $\epsilon$ i.e. $\epsilon^{-1/2}$, but only having the value $\gamma\lambda$ reduce the complexity and enable a quartic speedup in total performance.

In practical scenarios, both $|A\rangle$ and $|B\rangle$ can deviate from the above extreme case, making the improvement $\gamma\lambda$ milder. 
Also, when we consider the general setting where we use CBE to prepare $\rho_A$ and use LCU to prepare $U_B$ with $|B\rangle$ obeying the form of Eq.~\eqref{bform}, the value of $\gamma\lambda$ can be further reduced. Nevertheless, whenever $\||B\rangle\|_1$ is within $2^{\mathit{o}(n)}$ and the circuit depth (number of CNOT gates connecting US and LS) or the Hamiltonian simulation time ($\|H_{\text{Inter}}\|t$) for preparing $|A\rangle$ satisfies $n-K=\Theta(n)$ ($n/2-\|H_{inter}\|t=\Theta(n)$), the value of $\gamma\lambda$ is still exponentially large. Since the gate complexity between preparing $|B\rangle$ and constructing $U_B$ are equal \cite{zhang2022quantum}, which is also true for the gate complexity between preparing $|A\rangle$ and $\rho_A$ as shown in the construction of the CBE of individual gates, we can therefore conclude the regime where our algorithm has practical exponential improvements.

\section{Detailed discussion on implications and applications}

\subsection{$|A\rangle$ as a product state\label{ooru}}
When $|A\rangle=|A_u\rangle|A_l\rangle$ is a product state between US and LS, we can achieve the largest value of $\gamma=1/2$. For such $|A\rangle$, while there is no entanglement between US and LS, $|A\rangle$ can still be classically intractable 
because of the entanglement within US and LS. Also, since we generally require $|B\rangle$ to have a large entanglement between US and LS, estimating $\langle B|A\rangle$ is intrinsically a $2n$-qubit problem despite $|A\rangle$ as a product state and thus can not be efficiently turned down to $n$-qubit systems by circuit cutting methods \cite{harrow2024optimal,peng2020simulating}. 

$\textbf{BQP-completeness}$: If we have $|A\rangle=|A_u\rangle|A_u^*\rangle$ with $|A_u\rangle$ an arbitrary $n$-qubit state and set $|B\rangle=|B_b\rangle$ a Bell basis state, then we have $\langle B|A\rangle=2^{-n/2}\langle A_u|P|A_u\rangle$ with $P$ a $n$-qubit Pauli operator (Here, $A=|A_u\rangle\langle A_u|$). It is known that estimating the Pauli expectation value up to a (polynomial) small additive error is BQP-complete \cite{aharonov2017interactive,janzing2005ergodic} and the essence of our method is to use $\langle A_u|P|A_u\rangle$ to estimate $\langle B|A\rangle$ to a small relative error. When $\langle B|A\rangle$ is of order $\mathcal{O}(2^{-n/2})$, this relative error directly corresponds to the small additive error of estimating Pauli expectation values, and thus the whole setup is also BQP-complete. This also indicates that the efficient estimation of $\mathcal{O}(2^{-n/2})$ amplitudes in Table \ref{tbl1} is the best we can expect and it is beyond the capability of quantum computers for smaller amplitudes such as those of order $\mathcal{O}(2^{-n})$.

$\textbf{Bell sampling}$: Recently, several protocols have been developed, ranging from quantum learning \cite{huang2021information,huang2022quantum} to quantum sampling advantage \cite{hangleiter2024bell}. These protocols are based on a setup where $|A\rangle = |A_u\rangle |A_u\rangle$, with $|A_u\rangle$ representing an arbitrary $n$-qubit state, and $|B\rangle = |B_b\rangle$, is a Bell basis state. In Ref. \cite{hangleiter2024bell}, the authors propose the Bell sampling protocol and show that any GapP function can be encoded into the amplitude $\langle B|A\rangle$. Estimating $\langle B|A\rangle$ to a small relative error corresponds to estimating the values of GapP functions to a small relative error, which is GapP-hard \cite{hangleiter2023computational}. Therefore, our method is unable to estimate such amplitudes efficiently. Otherwise, quantum computers could solve problems across the entire polynomial hierarchy \cite{toda1991pp}. This also indicates that such amplitudes should be exponentially smaller than $\mathcal{O}(2^{-n/2})$ such that the resulting complexity of our method is still exponential (as shown in Table \ref{tbl1}), which coincides with quantum sampling advantage protocols where the anti-concentration effect \cite{dalzell2022random,hangleiter2023computational} makes almost all amplitudes around $\mathcal{O}(2^{-n})$.

\subsection{$|A\rangle$ as a entangled state \label{bbbbbbbb}}
When $|A\rangle$ is prepared by a quantum circuit or a Hamiltonian simulation of depth (time) with $n-K=\Theta(n)$ ($n/2-\|H_{inter}\|t=\Theta(n)$), $\gamma\lambda$ is still exponential when $B$ is proportional to a unitary operator. This regime is interesting because we can see this setup as a genuine $2n$-qubit setup. 

When $|A\rangle$ is a product state (Subsection \ref{ooru}), while the large entanglement in $|B\rangle$ ensures the amplitude is a $2n$-qubit amplitude, we can see from above that we can use matrixization to convert it into a $n$-qubit value. Thus, the exponential improvement of our method can be seen as finding a way or a picture to turn these amplitudes into what they really are in terms of hardness. In other words, they are essentially $n$-qubit values but hidden in $2n$-qubit amplitudes.

However, this $2n$ to $n$ conversion doesn't hold for entangled states. This can be understood from four facts. 
First, from the tensor network picture, the circuit or the Hamiltonian simulation allows a linear depth or a linear evolution time that has an exponential large bond dimension which is beyond the efficiently separable regime \cite{cirac2021matrix,peng2020simulating}. Second, due to the interactions between US and LS, we can no longer use matrixization to convert it into a straightforward $n$-qubit value as in the above two cases. Third, when $|B\rangle=U_{clif}|0\rangle$ is a stabilizer state, the amplitudes can be understood as the amplitudes of states prepared by Clifford circuits with non-stabilizer inputs ($|A\rangle$) projecting on the computational basis \cite{yoganathan2019quantum,gottesman1998heisenberg}. Fourth, such depth/simulation time is already enough for the emergence of various interesting things such as the approximate unitary design \cite{schuster2024random}, pseudorandom unitaries \cite{schuster2024random}, and anti-concentration \cite{dalzell2022random}.

Thus, this not maximally entangled case is a genuine $2n$-qubit setup, and our method is able to truly give a exponential improvement over previous methods. This doesn't violate the no-go theorem set by the linearity of quantum mechanics \cite{bennett1997strengths,childs2016optimal,abrams1998nonlinear} since to build CBE for $V$, we need additional information about how $V$ is constructed rather than a black box. Even though, to the best of our knowledge, there are no known results utilizing this additional information to go beyond quadratic speedup ($\mathcal{O}(|\mu|^{-1})$ complexity). Our method can achieve this because the logic of our method is not to estimate $\langle B|A\rangle$ directly but to first convert $|A\rangle$ and $|B\rangle$ into other objects by DMSE and UBSE and then do the estimation. 

\subsection{Hardness of $|B\rangle$ \label{ape9}}
We have shown that if $B$ is proportional to a unitary operator $U_B$, then this unitary operator is exactly a UBSE of $B$ with the largest $\lambda=2^{n/2}$. If we further set $|A\rangle=|A_u\rangle|A_l\rangle$, then we have $\langle B|A\rangle=\langle A_l^*|U_B^\dag |A_u\rangle$. Since $U_B$ can be an arbitrary unitary operator, we can even set $|A\rangle=|0\rangle|0\rangle$ and still encode any GapP functions into the amplitude $\langle B|A\rangle$ \cite{hangleiter2023computational}. Thus, the argument is similar to the above discussions on Bell sampling.

We can also consider the average behaviors of $B$ and $|B\rangle$. If $|B\rangle$ is picked from a state 2-design \cite{mele2024introduction},
with a high probability, $|B\rangle$ is close to a maximally entangled state that meets our algorithm,
since the average purity of the reduced density matrix of $|B\rangle$ is $2^{n+1}/(2^{2n}+1)\approx 2^{-n}$. 
Since when $B$ is proportional to a unitary operator, $|B\rangle$ is exactly a maximally entangled state, thus, a natural guess is that the randomness of $B$ can be inherited by $|B\rangle$. Indeed, we prove that when the $n$-qubit operator $U_B$ is drawn from a unitary 2-design, the ensemble $\{|B\rangle=2^{-n/2}\ket{\ket{U_B}}\}$ forms an $\mathcal{O}(2^{-2n})$-approximate $2n$-qubit state 2-design in terms of the trace distance. 
\begin{theorem}
When the $n$-qubit operator $U_B$ is drawn from a unitary 2-design, the ensemble $\{|B\rangle=2^{-n/2}\ket{\ket{U_B}}\}$ forms an $\mathcal{O}(2^{-2n})$-approximate $2n$-qubit state 2-design in terms of the trace distance
\end{theorem}
\begin{proof}
First, we have $|B\rangle=(U_B\otimes I_n)|B_I\rangle$ with $|B_I\rangle=2^{-n/2}\sum_i |i\rangle|i\rangle$. Thus, following the Weingarten calculus \cite{mele2024introduction}, we obtain:
\begin{eqnarray}
&&\int_{U_B\sim \nu} (U_B\otimes I_2)|B_I\rangle\langle B_I| (U_B^\dag\otimes I_4))^{\otimes 2}\nonumber\\&&=\frac{Tr_{13}(|B_I\rangle\langle B_I|\otimes |B_I\rangle\langle B_I|)-2^{-n}Tr_{13}((\text{SWAP}_{13}\otimes I_{24} )(|B_I\rangle\langle B_I|\otimes |B_I\rangle\langle B_I|))}{2^{2n}-1}\otimes I_{13}\nonumber\\&&+\frac{Tr_{13}((\text{SWAP}_{13}\otimes I_{24})(|B_I\rangle\langle B_I|\otimes |B_I\rangle\langle B_I|))-2^{-n}Tr_{13}(|B_I\rangle\langle B_I|\otimes |B_I\rangle\langle B_I|)}{2^{2n}-1}\otimes\text{SWAP}_{13}\nonumber\\&&=\frac{2^{-2n}}{2^{2n}-1}I_{1234}-\frac{2^{-3n}}{2^{2n}-1}I_{13}\otimes\text{SWAP}_{24}+\nonumber\\&&\frac{2^{-2n}}{2^{2n}-1}\text{SWAP}_{13}\otimes\text{SWAP}_{24}-\frac{2^{-3n}}{2^{2n}-1}\text{SWAP}_{13}\otimes I_{24}\nonumber,
\end{eqnarray}
where $\nu$ is a unitary 2-design and the four $n$-qubit systems in $|B\rangle |B\rangle$ are labeled as $1$, $2$, $3$, and $4$ in order. We can compare this result with the symmetric subspace projector (the definition of state 2-design) in terms of the trance distance:
\begin{eqnarray}
&&\left\|\int_{U_B\sim \nu} (U_B\otimes I_2)|B_I\rangle\langle B_I| (U_B^\dag\otimes I_4))^{\otimes 2}-\frac{I_{1234}+\text{SWAP}_{13}\otimes\text{SWAP}_{24}}{2\binom{2^{2n}+1}{2}}\right\|_1\nonumber\\&&=\left\|\frac{2^{2n+1}}{2^{8n}-2^{4n}}(I_{1234}+\text{SWAP}_{13}\otimes\text{SWAP}_{24})-\frac{2^{-3n}}{2^{2n}-1}I_{13}\otimes\text{SWAP}_{24}-\frac{2^{-3n}}{2^{2n}-1}\text{SWAP}_{13}\otimes I_{24}\right\|_1\nonumber\\&&\leq \frac{2^{2n+1}}{2^{4n}-1}+\frac{2^{2n+1}}{2^{6n}-2^{2n}}+\frac{1}{2^{2n}-1}+\frac{1}{2^{2n}-1}=\mathcal{O}(2^{-2n})
.\end{eqnarray}
Thus, $|B\rangle$ indeed forms an approximate state 2-design.
\end{proof}
\noindent Since 2-design is a sufficient condition for the anti-concentration effect \cite{hangleiter2018anticoncentration}, the amplitudes under this setting are on average of order $\mathcal{O}(2^{-2n})$.

\subsection{Probing properties of low-temperature Gibbs states from high-temperature ones \label{ape10}}
An intriguing implication of the state matrixization framework is that it allows us to investigate certain properties of low-temperature Gibbs states using high-temperature ones. The motivation is that $|A\rangle$ and $\rho_A$ are essentially different states, thus, there may exist a complexity separation in terms of their preparations. Indeed, if $\rho_A$ is prepared by CBE, we have shown its preparation cost is at the same level as $|A\rangle$. However, for Gibbs state, we can prepare $\rho_A$ without CBE. Since an $n$-qubit Gibbs state $\rho_\beta=\frac{e^{-\beta H}}{\Tr(e^{-\beta H})}$ is exactly a DMSE of the purified Gibbs state of $\rho_{2\beta}$. Now suppose we want to probe properties such as $\langle B|\rho_\beta\rangle$ with $|\rho_\beta\rangle$ denotes the normalized vectorized state from $\rho_\beta$, in the traditional method, we need to actually prepare $|\rho_\beta\rangle$ which corresponds to the preparation of the Gibbs state at the temperature $1/(2\beta)$ which can be hard for quantum computers \cite{lucas2014ising,kempe2006complexity}, in contrast, using our method, we only need to prepare a Gibbs state at the higher temperature $1/\beta$ which can be easy for quantum computers \cite{bakshi2024high}. 

Here, we give an example where the query complexity of estimating $\langle B|\rho_\beta\rangle$ is comparable for both methods but there can be an exponential separation on the preparation complexity between $\rho_\beta$ and $\rho_{2\beta}$. 

If we have a $n$-qubit Gibbs state $\rho_\beta=\frac{e^{-\beta H}}{\Tr(e^{-\beta H})}$ at the temperature $1/\beta$, then its vectorization is an un-normalized purified Gibbs state:
\begin{eqnarray}
\mathcal{V}[\rho_\beta]=\frac{e^{-\beta  \ket{\ket{H}}}}{\Tr(e^{-\beta H })}=\frac{\sqrt{\Tr(e^{-2\beta H})}}{\Tr(e^{-\beta H })}|\rho_\beta\rangle
,\end{eqnarray}
where we have:
\begin{eqnarray}
Tr_l(|\rho_\beta\rangle\langle \rho_\beta|)=\rho_{2\beta}
,\end{eqnarray}
with $Tr_l(\cdot)$ denotes the partial trace on the LS. Thus, $\rho_\beta$ itself is exactly a DMSE of the purified Gibbs state of $\rho_{2\beta}$ with no need of adding ancilla qubits. Now suppose we want to probe properties such as $\langle B|\rho_\beta\rangle$, in the traditional method, we need to actually prepare $|\rho_\beta\rangle$ which corresponds to the preparation of the Gibbs state at the temperature $1/(2\beta)$, in contrast, using our method, we only need to prepare a Gibbs state at the higher temperature $1/\beta$. In this case, we have $\gamma=\frac{\sqrt{\Tr(e^{-2\beta H})}}{\Tr(e^{-\beta H })}$, if $|B\rangle$ is further a maximally entangled state, we then have $\gamma\lambda=2^{n/2}\frac{\sqrt{\Tr(e^{-2\beta H})}}{\Tr(e^{-\beta H })}$.

We can consider an interesting case where $\rho_\beta$ has the spectral $\{p_0,p_1,p_1,... ,p_1\}$ with $p_0=2^{-n/2}$ and $p_1=(1-2^{-n/2})/(2^n-1)$. Then $\rho_{2\beta}$ has the spectra $\{p'_0,p'_1,p'_1,... ,p'_1\}$ with:
\begin{eqnarray}
p'_0&&= \frac{2^{-n}(2^n-1)}{2^{-n}(2^n-1)+(1-2^{-n/2})^2}=\mathcal{O}(1),\nonumber\\
p'_1&&=\frac{(1-2^{-n/2})^2}{2^{-n}(2^n-1)^2+(2^n-1)(1-2^{-n/2})^2}=\mathcal{O}(2^{-n}).
\end{eqnarray}
In this case, we have:
\begin{eqnarray}
\gamma\lambda=2^{n/2}\frac{\sqrt{\Tr(e^{-2\beta H})}}{\Tr(e^{-\beta H })}=2^{n/2}\frac{\sqrt{\frac{e^{-2\beta E_0}}{p_0'}}}{\frac{e^{-\beta E_0}}{p_0}}=\frac{\sqrt{(2^n-1)+2^n(1-2^{-n/2})^2}}{\sqrt{2^n-1}}=\mathcal{O}(1)
.\end{eqnarray}
Thus, in terms of the query complexity of estimating $\langle B|\rho_\beta\rangle$, we are comparable with traditional direct measurement. However, the preparation complexity can have an exponential separation between $\rho_\beta$ and $\rho_{2\beta}$. The reason is that in $\rho_{2\beta}$, the ground state population is of order $\mathcal{O}(1)$ which means the preparation of $\rho_{2\beta}$ has the same complexity as the preparation of the ground state of $H$ which should require an exponential amount of time when estimating the ground energy of $H$ is an NP-hard instance \cite{lucas2014ising} or a QMA-hard instance \cite{kempe2006complexity}. On the other hand, in $\rho_\beta$, the ground state population is still exponentially small, thus, the preparation of $\rho_\beta$ cannot be modified to the preparation of the ground state. Therefore, we may expect the existence of a polynomial-time preparation of $\rho_\beta$. Also, since estimating of $\langle B|\rho_\beta\rangle$ doesn't correspond to estimating the ground energy, the whole setup has no violations of any complexity beliefs.

\end{appendix}
\end{document}